\documentclass[letterpaper,conference,compsoc]{IEEEtran}
\def\version{SP}
\newif \ifsubmission \submissionfalse

    \newif \iffull 
\newif \ifACM
\newif \ifUSENIX
\newif \ifIEEE
\newif \ifLNCS

\newif \ifCCS
\newif \ifSP
\newif \ifNDSS
\newif \ifCrypto
\newif \ifFC

\def\fullstring{full}
\def\ACMstring{ACM}
\def\USENIXstring{USENIX}
\def\IEEEstring{IEEE}
\def\LNCSstring{LNCS}

\def\CCSstring{CCS}
\def\SPstring{SP}
\def\NDSSstring{NDSS}
\def\Cryptostring{CRYPTO}
\def\FCstring{FC}

\ifx \version\fullstring \fulltrue \fi
\ifx \version\ACMstring \ACMtrue \fi
\ifx \version\USENIXstring \USENIXtrue \fi
\ifx \version\IEEEstring \IEEEtrue \fi
\ifx \version\LNCSstring \LNCStrue \fi

\ifx \version\CCSstring \ACMtrue \CCStrue \fi
\ifx \version\SPstring \IEEEtrue \SPtrue \fi
\ifx \version\NDSSstring \IEEEtrue \NDSStrue \fi
\ifx \version\Cryptostring \LNCStrue \Cryptotrue \fi
\ifx \version\FCstring \LNCStrue \FCtrue \fi

\ifACM \input{Conferences/ACM/acm.tex} \fi
\ifUSENIX \input{Conferences/USENIX/usenix.tex} \fi
\ifIEEE \errorcontextlines=200
\usepackage[noadjust]{cite}
\usepackage[scaled=0.88]{helvet}

 \fi
\ifLNCS \input{Conferences/LNCS/lncs.tex} \fi

\newif \ifcomments \commentsfalse
\newif \ifanon \anonfalse
\newif \ifwipnocirc \wipnocircfalse

\ifUSENIX \else \usepackage[table]{xcolor} \fi
\usepackage{xurl}
\usepackage{xspace}
\usepackage{hyperref}
\usepackage{amsmath}
\hypersetup{
    colorlinks=true,
    allcolors=blue,
    linkcolor=red,
    filecolor=magenta,
}
\usepackage[many]{tcolorbox}        
\usepackage{lipsum}
\usepackage{cleveref}
\usepackage{enumitem}
\usepackage{booktabs}
\usepackage{enumitem}
\usepackage{makecell}
\usepackage{multirow}
\usepackage[figuresright]{rotating}
\usepackage{array}
\usepackage{nicematrix}
\usepackage{tikz}
\usepackage{subcaption}
\usepackage{booktabs}
\usepackage{siunitx}
\usepackage{algorithm}
\usepackage{algpseudocode}
\usepackage[cal=cm]{mathalfa}
\usepackage{mdframed}
\usepackage{comment}
\usepackage{tcolorbox}
\usepackage{diagbox}
\usepackage{eso-pic}
\usepackage{xcolor}

\setlist[itemize]{itemsep=0.2em, parsep=0pt, topsep=0.3em}

\usepackage{titlesec}
\titleformat{\paragraph}[runin]{\normalfont\bfseries}{\theparagraph}{0.75em}{}
\titlespacing*{\paragraph}{0pt}{1ex}{0.75em} %

\algtext*{EndFor}
\algtext*{EndIf}

\usetikzlibrary{shapes, arrows.meta, positioning, calc, math, tikzmark, fit}

\usepackage{amsfonts}
\ifLNCS\else\usepackage{amsthm}\fi
\ifACM\else\usepackage{amssymb}\fi

\newcommand{\protbox}[2]{\fbox{\small\hbox{\begin{minipage}{0.97\columnwidth}\begin{center}{\bf #1}\end{center}#2\end{minipage}}}}

\definecolor{ForestGreen}{RGB}{34,139,34}

\ifcomments
    \newcommand{\mahimna}[1]{\textsf{\small{\color{violet!80}{[Mahimna: {#1}]}}}}
    \newcommand{\kushal}[1]{\textsf{\small{\color{blue}{[Kushal: {#1}]}}}}
    \newcommand{\james}[1]{\textsf{\small{\color{green!75!black}{[James: {#1}]}}}}
    \newcommand{\ari}[1]{\textsf{\small{\color{red}{[Ari: {#1}]}}}}
    \newcommand{\jay}[1]{\textsf{\small{\color{orange}{[Jay: {#1}]}}}}
    \newcommand{\sarah}[1]{\textsf{\small{\color{red}{[Sarah: {#1}]}}}}
    \newcommand{\andres}[1]{\textsf{\small{\color{blue}{[Andres: {#1}]}}}}
    
    \newcommand{\dani}[1]{\textsf{\small{\color{purple}{[Dani: {#1}]}}}}
    \newcommand{\sam}[1]{\textsf{\small{\color{yellow!75!black}{[Sam: {#1}]}}}}
    \newcommand{\paddy}[1]{\textsf{\small{\color{pink}{[Paddy: {#1}]}}}}
\else
    \newcommand{\kushal}[1]{}
    \newcommand{\mahimna}[1]{}
    \newcommand{\james}[1]{}
    \newcommand{\ari}[1]{}
    \newcommand{\jay}[1]{}
    \newcommand{\sarah}[1]{}
    \newcommand{\andres}[1]{}
    \newcommand{\dani}[1]{}
    \newcommand{\sam}[1]{}
    \newcommand{\paddy}[1]{}
\fi

\tcbset{
    sharp corners,
    colback = white,
    before skip = 0.2cm,    
    after skip = 0.5cm      
}                           %

\newtcolorbox{boxA}{
    fontupper = \bf,
    boxrule = 1.5pt,
    colframe = black %
}

\newcommand{\mypara}[1]{\smallskip\noindent\textbf{#1}\;}

\ifLNCS\else

\newtheorem{theorem}{Theorem}

\newtheorem{definition}[theorem]{Definition}

\theoremstyle{remark}

\fi

\definecolor{keyFindingColor}{HTML}{4A90E2}   %
\definecolor{keyFindingBackground}{HTML}{e9edf5} %

\usepackage[many]{tcolorbox}
\usepackage{enumitem}

\newlist{keyfindings}{itemize}{1}
\setlist[keyfindings,1]{%
  label=\(\triangleright\), %
  leftmargin=1.4em,
  itemsep=2pt plus 1pt minus 1pt,
  topsep=4pt plus 1pt minus 1pt
}

\newtcolorbox{keyFinding}[1][]{%
  enhanced, breakable,
  colback=keyFindingBackground,
  colframe=keyFindingColor,
  coltitle=black,
  boxrule=0pt, frame hidden,
  borderline west={3pt}{0pt}{keyFindingColor},
  sharp corners,
  left=10pt,right=10pt,top=8pt,bottom=10pt,
  before skip=8pt plus 2pt minus 2pt,
  after  skip=10pt plus 2pt minus 2pt,
  fonttitle=\bfseries,
  title=Key Finding, %
  attach boxed title to top left={yshift=-2mm, xshift=8pt},
  boxed title style={colback=white,colframe=keyFindingColor,boxrule=0.5pt,sharp corners,
                     left=8pt,right=8pt,top=2pt,bottom=2pt},
  before upper=\vspace{4pt},
  #1 %
}

\ifCCS
\setcopyright{acmlicensed}
\copyrightyear{2018}
\acmYear{2018}
\acmDOI{XXXXXXX.XXXXXXX}
\acmConference[Conference acronym 'XX]{Make sure to enter the correct conference title from your rights confirmation email}{June 03--05, 2018}{Woodstock, NY}
\acmISBN{978-1-4503-XXXX-X/2018/06}
\fi

\begin{document}

\title{Crossroads: A Smart Contract Layer for Chain-Abstracted Assets}

\ifSP
\IEEEoverridecommandlockouts
\makeatletter
\newcommand{\linebreakand}{%
  \end{@IEEEauthorhalign}
  \hfill\mbox{}\par
  \mbox{}\hfill\begin{@IEEEauthorhalign}
}
\makeatother
\fi
\author{
	\IEEEauthorblockN{James Austgen\textsuperscript{*}}
	\IEEEauthorblockA{Cornell Tech, IC3\\
		\href{mailto:james@cs.cornell.edu}{james@cs.cornell.edu}}
	\and
	\IEEEauthorblockN{Dani Vilardell\textsuperscript{*}}
	\IEEEauthorblockA{Cornell Tech, IC3\\
		\href{mailto:danivilardell@cs.cornell.edu}{danivilardell@cs.cornell.edu}}
    \and
	\IEEEauthorblockN{Ari Juels}
	\IEEEauthorblockA{Cornell Tech, IC3\\
		\href{mailto:juels@cornell.edu}{juels@cornell.edu}}
    \thanks{\textsuperscript{*}These authors contributed equally to this work.}
}

\ifACM
\renewcommand{\shortauthors}{Austgen et al.}
\ifACM\else
\renewenvironment{abstract}
  {\par\noindent\textit{\textbf{Abstract—}}\ignorespaces}
  {\par}
\fi

\begin{abstract}
This paper introduces \textbf{Crossroads}, a smart contract layer for \textit{chain-abstracted assets}. In Crossroads, assets from nearly \emph{any} chain are represented on a single backend blockchain as ERC-20 tokens. 
As a result, any asset can participate in smart-contract-based exchange, lending, or privacy applications on a single unified platform. So while Crossroads offers cross-chain bridging, a common, partial approach to alleviating the fragmentation of the blockchain ecosystem today, this is just one service within Crossroads' general-purpose chain-abstraction model.

Crossroads relies on \emph{key encumbrance}: a threshold signing committee holds encumbered keys controlling assets on each integrated chain, signing transactions only as authorized by smart contracts on the backend blockchain. Asset movements are fee-efficient, as ownership changes are recorded on the backend blockchain and users may set the transaction fee for withdrawals.

Crossroads enables permissionless, modular integration of new blockchains using pluggable oracles with flexible design options (zkBridge, TEE-based, hybrid). Asset deposits into Crossroads benefit from strong, chain-specific finalization guarantees, minimizing the risk of reorg attacks. Unlike existing bridges, however, third-party smart contracts in Crossroads can provide fast, optimistic access to funds before finalization completes. 

We prove that Crossroads satisfies soundness: given an honest quorum of signing committee members, any user can unilaterally generate a valid withdrawal transaction transferring their full net balance to an externally owned account on an integrated blockchain. We implement a proof of concept across multiple public blockchains: Bitcoin, Ethereum, and Solana. We catalog a range of applications enabled by Crossroads, including universal wallets, cross-chain staking and lending, privacy-preserving payments, and private management of public blockchain assets.
\end{abstract}



\maketitle
\else \maketitle  \pagestyle{plain} \fi

\sloppy

\section{Introduction}

Cross-chain bridges are the main path today towards blockchain interoperability. Current %
bridges, though, generally only support a handful of high-liquidity chains or chains within a single family (EVM, Cosmos, Solana, Bitcoin, etc.). Swaps between two cryptocurrencies often have to traverse several cross-chain exchanges, e.g., an exchange of StarkNet ETH to Solana SOL, might involve three or more bridge operations and several intermediary token swaps~\cite{rangoAPI}.

The challenge is that a single bridge connects exactly two chains. Full pairwise connectivity is impractical, as supporting $N$ chains would require $N^2$ bridges, each separately deployed and operated. Practical designs today therefore employ a sparse graph of bridges between high-liquidity chains, where transfers between unconnected chains must hop through intermediaries and pay fees at each step. 

A hub topology avoids such hopping: $N$ connections make every pair of chains reachable through the hub, and $N$ bridges are easier to audit than $N^2$. Existing hub-and-spoke systems, however, fall short of this potential in three ways: integrating a new chain is a permissioned decision of the hub's operators or governance; the hub's functionality is fixed at the protocol level, typically transfers or swaps, rather than open to third-party applications; and integration generally requires verification logic on the connected chain itself, leaving chains without smart contracts, such as Bitcoin, behind.  

\smallskip
\mypara{Crossroads.}
In this work, we present Crossroads, a smart contract layer for chain-abstracted assets. Crossroads adopts a hub-and-spoke topology but, unlike existing hubs, supports permissionless integration of any chain and enables a programmable layer that lets arbitrary smart contracts operate on the assets within the hub. Crossroads works with nearly any blockchain, even non-EVM chains and chains without smart contracts. It abstracts all blockchain assets as ERC-20 tokens existing on a single EVM-based chain (the \emph{backend chain}). On this backend chain, the exchange and transfer of cross-chain assets operate as ordinary on-chain transactions.

To interact with Crossroads-enabled cross-chain applications, users deposit assets into a committee controlled address on the asset's native chain; a per-chain oracle then attests on the backend chain that the deposit transaction is finalized, after which the corresponding ERC-20 tokens are minted on the backend chain. The committee holds the keys to these addresses under \emph{key encumbrance}~\cite{austgen2025liquefaction}: keys are generated and stored under a threshold signature scheme, so no party ever learns them, and signatures are produced only as authorized by the asset contracts on the backend chain. 

Concretely, to withdraw an asset from the Crossroads system, a user first burns its corresponding ERC-20 tokens on the backend chain. The committee verifies the burn against the asset contract and signs a transaction transferring the corresponding native assets to an address chosen by the user, with the user setting the fee and deciding when to broadcast.

Compared to existing cross-chain interoperability platforms, this design gives Crossroads three novel properties. First, cross-chain functionality is programmable: any smart contract deployed on the backend chain (an AMM, a lending market, a stablecoin issuer) operates on assets from every integrated chain. Because each ERC-20 token carries the right to withdraw its underlying asset, transferring tokens also transfers that right with them; a single backend transaction, such as an AMM swap, thus settles an exchange of native assets atomically. Second, integration is permissionless: integrating a chain amounts to deploying its ERC-20 asset contract on the backend chain and an oracle relaying the chain's finalized transactions, so anyone can add a chain or asset without the approval of a developer or system operator. Third, the requirements on integrated chains are minimal: because Crossroads needs only a unidirectional oracle from the asset chain to the backend chain, any chain with eventual, verifiable finality can be integrated, with no contracts, deployed logic, or protocol changes on the chain itself. This is what brings chains without smart contract support, such as Bitcoin and XRPL, within reach.

Crossroads' security rests on the standard trust assumption of threshold cryptosystems: an adversary corrupts fewer signing committee members than the signing threshold. Under this assumption, we prove that Crossroads is \textit{sound}: any user can redeem their Crossroads-issued ERC-20 tokens for the underlying asset and transfer them to an externally owned account (EOA) on the asset's native chain. Soundness implies that an adversary in our threat model cannot steal or lose user funds. In view of committee compromises in practice~\cite{krause20251,chainalysis2026kelpdao,zhang2024cross}, however, we present modular security enhancements that provide defense in depth should our corruption assumption fail. These include use of secure hardware (TEEs), security councils with freeze authority, and accountable threshold signatures that identify misbehaving signers.

Crossroads' design ideas have already seen industry adoption. Oasis Sapphire, a TEE-based blockchain, has integrated many of Crossroads' components into its multi-chain programmable wallet SDK, Privana~\cite{oasis2026privana}.

\smallskip

\mypara{Applications.} Because Crossroads represents assets from arbitrary chains as ERC-20 tokens on a single programmable backend, any ERC-20 compatible smart contract on the backend  automatically becomes a cross-chain application. This can turn systems which otherwise would require bridging infrastructure into ordinary single-chain smart contracts.

One direct application of Crossroads is a \emph{chain-agnostic wallet}, a single key pair that controls assets across every integrated chain, with the user’s full portfolio visible on a single ledger and tokens spendable on any integrated chain without the need to hold native accounts there. Crossroads also allows standard exchange contracts, such as AMM, order book, or RFQ contracts, to be deployed onto the Crossroads platform. Lending and borrowing applications are similarly possible: collateral deposited from one chain can back loans denominated in assets from another, making cross-chain shorts and arbitrage practical with existing lending platforms such as Morpho~\cite{morpho} or Euler~\cite{euler}.

Crossroads also enables applications that existing bridges cannot offer at all. A \emph{universal stablecoin} can be issued natively on the backend and spent on any integrated chain, removing the centralized issuer’s role in deciding which chains to support. \emph{Cross-chain staking} via Crossroads offers a new, highly flexible, multi-chain means for emerging chains to bootstrap their security with established assets such as ETH or BTC, rather than just their native liquidity~\cite{tas2022babylon,tas2023bitcoin}. When the backend blockchain is itself privacy-preserving (e.g. Oasis Sapphire~\cite{cheng2019ekiden, oasis2023sapphire}), Crossroads further enables \emph{private payments} and \emph{private institutional management of public-chain assets}, with smart-contract-enforced disclosure and compliance policies layered directly onto exchange transactions.

\subsection*{Contributions}

\noindent We provide background to Crossroads in~\Cref{sec:background}. Subsequently, our contributions are as follows:

\begin{itemize}
    \item \textit{Crossroads:} We introduce Crossroads, a smart-contract layer for chain-abstracted assets. Crossroads can represent assets from \textit{any} chain as ERC-20s on its unified smart-contract platform (\Cref{sec:crossroads_arch}).
    \item \textit{Optimizations:} We present key optimizations to Crossroads' core design: per-user deposit addresses to enable broad wallet compatibility; optimistic deposits and withdrawals to reduce settlement latency; and parallelized withdrawals to increase throughput (\Cref{sec:practical}).
    \item \textit{Security:} We prove system soundness by showing that at any point a user can unilaterally withdraw their full net balance to an externally owned account. We also propose modular security enhancements to strengthen the security model (\Cref{sec:security}).
    \item \textit{Implementation:} We implement a proof-of-concept of Crossroads and integrate it with Bitcoin, Ethereum, and Solana (\Cref{sec:implementation}), available at {\url{https://github.com/trate3/crossroads}}.
    \item \textit{Applications:} We identify a broad range of applications for Crossroads, from chain-agnostic wallets and cross-chain decentralized exchanges to cross-chain staking and lending, privacy payments, and universal stablecoins (\Cref{sec:applications}).
\end{itemize}

We review related work in \Cref{sec:related-work} and conclude in \Cref{sec:conclusion}.

\section{Building Blocks}\label{sec:background}

We introduce the cryptographic and systems primitives that Crossroads builds on.

\subsection{Cross-chain oracles}

Blockchains cannot natively observe events outside their own ledger; \textit{oracles} are systems that make external events available to smart contracts. A \textit{cross-chain oracle}  lets a contract on a destination chain verify that an event---for our purposes, the inclusion and finality of a transaction---occurred on a source chain. 

Oracle constructions differ primarily in their trust model: \textit{committee-based} oracles rely on an honest majority of external attesters~\cite{wormhole, axelar}, \textit{light-client} oracles verify the source chain's consensus on the destination chain, possibly via zero knowledge proofs~\cite{xie2022zkbridge}, and \textit{TEE-based} oracles run prove transaction finality via their attestation capabilities. Crossroads requires  a unidirectional oracle of finalized asset-chain transactions into its backend chain, and any such system can serve that purpose, each with its own security guarantees.

\subsection{Cryptocurrency and DeFi}
Most major blockchains have a single \textit{native asset} which is required to pay for transaction fees and protect against unbridled resource consumption or denial of service. Many assets, however, do not require their own blockchain. These can be represented as \textit{tokens}, and on EVM chains, they are smart contracts which adhere to the ERC-20 standard~\cite{erc20standard}.

The popularity of ERC-20 tokens gave rise to \textit{decentralized finance} (DeFi), a term encompassing financial services implemented on smart contracts, such as token-exchange, lending, and derivative-asset platforms.

\subsection{Key encumbrance}
Nearly all blockchains use public-key cryptography for account authentication and transaction signing. Usually, individual users directly control their private key material. Under \textit{key encumbrance}, a policy, rather than an individual user, governs a key's usage. The policy acts as a proxy through which authenticated parties may obtain signatures from the key, and it decides who is authorized to obtain signatures for particular messages. The parties themselves never learn the key.

Key encumbrance can be used to construct a shared blockchain account in which parties who do not trust each other individually control their own portions of the account balance~\cite{austgen2025liquefaction}. Furthermore, a single key encumbrance platform can sign transactions for separate blockchains. Practically, a TEE application or MPC protocol is used to generate keys securely and enforce encumbrance policies over them.

\subsection{Threshold signature scheme}\label{sec:tss}
We require a $(t,N)$-threshold signature scheme comprising four algorithms $(\mathsf{KG}, \mathsf{Sign}, \mathsf{Verify}, \mathsf{Combine})$:
\begin{itemize}
    \item $\mathsf{KG}(1^\lambda, N, t) \to (\mathsf{pk}, \mathsf{pk}_c, \mathsf{sk}_1, \ldots, \mathsf{sk}_N)$: a probabilistic key generation algorithm that generates a $t$-out-of-$N$ shared key, outputting a public key $\mathsf{pk}$, a combiner public key $\mathsf{pk}_c$, and $N$ signing key shares $\mathsf{sk}_1, \ldots, \mathsf{sk}_N$.
    \item $\mathsf{Sign}(\mathsf{sk}_i, m) \to \sigma'_i$: a (possibly) probabilistic signing algorithm that takes a key share $\mathsf{sk}_i$ and a message $m$, and outputs a signature share $\sigma'_i$.
    \item $\mathsf{Verify}(\mathsf{pk}, m, \sigma) \to b$: a deterministic verification algorithm that takes a public key $\mathsf{pk}$, a message $m$, and a signature $\sigma$, and outputs $b \in \{\mathsf{accept}, \mathsf{reject}\}$.
    \item $\mathsf{Combine}(\mathsf{pk}_c, m, J, \{\sigma'_j\}_{j \in J}) \to \sigma$: a deterministic combiner algorithm that takes the combiner public key $\mathsf{pk}_c$, a message $m$, a subset $J \subseteq [N]$ with $|J| = t$, and signature shares $\{\sigma'_j\}_{j \in J}$, and either outputs a combined signature $\sigma$ or outputs $\mathsf{blame}(J^*)$ for some nonempty $J^* \subseteq J$, indicating that the shares $\{\sigma'_j\}_{j \in J^*}$ are invalid.
\end{itemize}
We require the scheme to satisfy correctness, robustness, and existential unforgeability under chosen message attack (EUF-CMA) against an adversary corrupting up to $t-1$ signers, as defined in~\cite{boneh2023applied}.

\paragraph{Distributed key generation.} Rather than relying on a trusted dealer to execute $\mathsf{KG}$ and distribute shares, Crossroads instantiates key generation through a secure \emph{distributed key generation} (DKG) protocol $\Pi_{\mathsf{DKG}}$ jointly run by the $N$ signing committee members, as defined in~\cite{boneh2023applied}.

\subsection{Trusted execution environments}
A trusted execution environment (TEE) is an isolated execution environment enabled by hardware design that protects applications running inside it from inspection or tampering by the host operating system, and in some designs from the machine owner as well. Programs executing within a TEE can produce outputs accompanied by a \emph{remote attestation}: a cryptographic statement, signed by a hardware-rooted key, that binds the output to the specific program that produced it. Any external party can then verify, against the manufacturer's attestation key, that the output was produced by an honest execution of the attested program without trusting the host machine.

Several commercial TEE implementations are widely deployed, including Intel SGX~\cite{intel_sgx}, AMD SEV-SNP~\cite{amd_sev_snp}, and AWS Nitro Enclaves~\cite{aws_nitro}. They differ in threat model, performance, and exposure to side-channel attacks. Recent work has demonstrated hardware-level vulnerabilities in deployed TEEs, including confidentiality leakage~\cite{wilke2024tdxdown, gast2025counterseveillance, yuan2025ciphersteal} and integrity attacks that enable forged attestations~\cite{RMPocalypse2025, chuang2026tee}. Crossroads' core security guarantees, however, do not rest on TEEs and we employ them only as an optional defense-in-depth layer.

\section{Crossroads Architecture}\label{sec:crossroads_arch}

In this section, we present the Crossroads architecture: its design goals and the components that fulfill them (\Cref{subsec:architecture-protocol}), requirements for asset support (\Cref{sec:requirements}), and the deposit and withdrawal flows (\Cref{subsec:architecture-deposits,subsec:architecture-withdrawals}).

\subsection{Protocol} 
\label{subsec:architecture-protocol}

\begin{figure*}
    \centering
    \includegraphics[width=\textwidth]{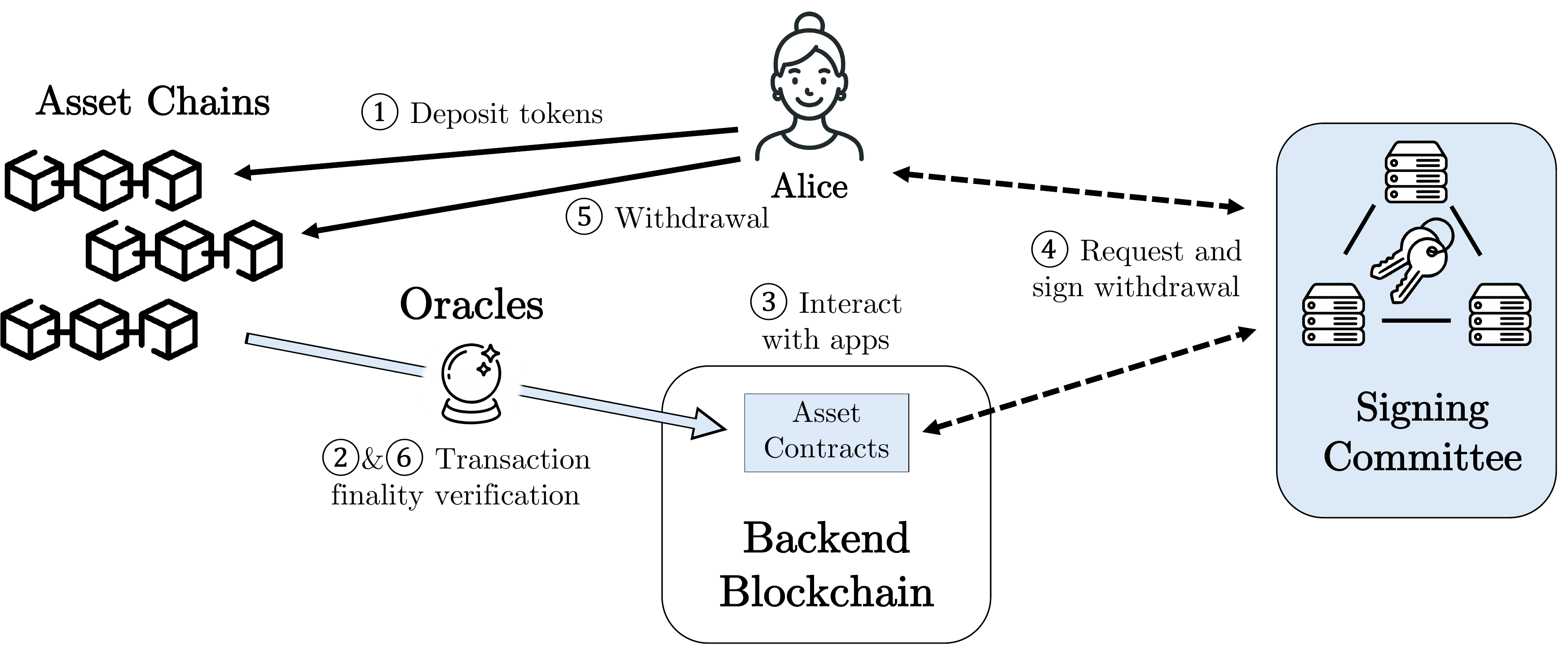}
    \caption{Overview of the Crossroads architecture. Alice deposits assets into encumbered addresses on asset chains (1); oracles verify the deposits and credit her balance on the backend blockchain (2), where she can interact with applications such as AMMs or staking contracts (3). To withdraw, Alice burns tokens on the backend and obtains a signed transfer transaction from the signing committee (4), broadcasts it to the asset chain (5), and the oracle verifies the withdrawal to settle her balance (6). Solid arrows denote on-chain transactions; dotted arrows, off-chain communication. The components in blue are core to Crossroads: the asset contracts, oracles, and the signing committee.} 
    \label{fig:arch-components}
\end{figure*}

Crossroads must solve several problems to achieve our goal of general cross-chain support. It must (1) hold assets on heterogeneous \textit{asset chains}, some without smart contracts; (2) support a representation of these assets on a programmable EVM \textit{backend blockchain}; and (3) maintain consistency between balances on the asset chains and the backend blockchain. Additionally, it must not require special permission to integrate new chains. New integrations must therefore add little additional overhead to the platform, and Crossroads itself must not be opinionated about the validity of a particular chain.

The first issue is how to control assets on separate chains. Nearly every blockchain handles spending the same way: a transaction is valid only if it carries a signature verifying under the account's key. So to custody a user's funds, Crossroads holds it in an account on the asset chain, and spending from that account requires its signing key. Rather than entrust that key to any single party, Crossroads splits it into shares held by a \textit{signing committee} that threshold-signs transactions or other messages.

We place all chain-specific logic in the backend blockchain, where each asset is represented as an ERC-20 \textit{asset contract}, ultimately constructing a hub-and-spoke cross-chain topology. The signing committee responds to signing requests sent over the network but only signs transactions if an asset contract has authorized the requester to obtain a signature for the requested message. This simple policy keeps the signing committee's logic small and separately auditable. 
This signing style is known as \emph{key encumbrance}---when meeting the requirements of the asset contract, a requester may request transaction signatures from the committee and broadcast the signed transactions directly, as if the owner of the (encumbered) account.

Crossroads keeps the backend balances blockchain consistent with the asset chain with the help of per-chain transaction \textit{finalization oracles}. Notably, these oracles only must report the state of the asset chain to the backend chain, so smart contracts are not required on asset chains. Crossroads instead relies on asset chains' own double-spending protection and transaction validity rules.

We now expand on the five main components within Crossroads. The three drawn in blue in~\Cref{fig:arch-components} are specific to Crossroads, while the asset chains and backend blockchain are existing systems.

\mypara{Asset chains.} Crossroads can integrate any blockchain that provides eventual, verifiable finality, including blockchains without smart contracts such as Bitcoin and XRPL, using a per-chain asset contract. The signing committee derives the signature for an encumbered address $\textsf{chain}_\textsf{addr}$ linked to that asset contract. Users deposit the assets into Crossroads by transferring them into $\textsf{chain}_\textsf{addr}$. While assets remain custodied at $\textsf{chain}_\textsf{addr}$ on the asset chain, all exchange and transfer activity takes place on the backend blockchain.

\mypara{Backend blockchain.} The backend blockchain is the programmable EVM chain where the Crossroads asset contracts are deployed alongside any other third-party smart contracts. Smart contracts compatible with ERC-20 tokens can therefore operate on assets from every integrated asset chain. For instance, an AMM contract such as Uniswap~\cite{uniswapv3} would allow Crossroads tokens from one asset chain to be exchanged atomically for those from another. The backend blockchain can similarly host lending markets, staking platforms, and privacy-preserving payments, all built using standard tooling without any awareness of the underlying asset chains.

\mypara{Asset contracts.}
Every blockchain asset integrated into Crossroads is represented by a dedicated asset contract on the backend blockchain. These contracts follow the ERC-20 standard and track each user's holdings of wrapped tokens. Each asset contract mediates all operations that affect user balances within Crossroads, including deposits, transfers, and withdrawals. It also exposes a view function called \texttt{canSign} to the signing committee, allowing transaction signing requests to be checked against recorded balances. The full asset contract functionality is presented in Figure~\ref{fig:crossroads-program}. 

\mypara{Oracles.} A transaction finalization oracle informs the backend blockchain when a transaction is finalized on the asset chain, letting the asset contracts verify a deposit or withdrawal. Crossroads supports flexible oracle designs, including zkBridges, TEE-based oracles, and hybrid approaches. The choice of oracle determines the security guarantees for each asset contract (\Cref{sec:security}); a single chain may even be served by multiple asset contracts each backed by different oracles.

\mypara{Signing committee.} The signing committee jointly holds the keys controlling Crossroads assets on each asset chain under key encumbrance. Shares of each key are distributed across committee members via a threshold signature scheme, so no party ever learns a key. Users communicate with the signing committee off-chain, and spending from a Crossroads-controlled address requires a threshold of members to cooperate. The encumbrance policy is read from the asset contract: members produce signature shares only for requests the asset contract authorizes, so withdrawal transactions are only signed when the corresponding tokens are locked. We specify the committee protocol in Figure~\ref{fig:signing-committee}.

\vspace{5mm}

The high-level architecture of Crossroads and the associated data flow for users are depicted in~\Cref{fig:arch-components}.

Each asset represented in Crossroads contains its own blockchain-specific deposit and withdrawal mechanisms. As Crossroads aims to be widely permissive in the blockchains it supports, we make few assumptions about blockchain architecture. In the following sections, we detail these assumptions and explain how deposits into and withdrawals from the Crossroads network operate.

\begin{figure}[!t]
\protbox{Crossroads Assets Contract $\textsf{crsrds}(\mathcal{O}, \textsf{TB})$}
{
    \textit{Deposits:} $D \gets \emptyset$\\
    \textit{Balances:} $\textsf{bal} \gets \{\, (\textsf{acct},\textsf{chain}) \mapsto 0 \mid \textsf{acct} \in \mathcal{U},\ \textsf{chain} \in B \cup \{\eta\} \,\}$\\
    \textit{Reward:} $r = \bot$\\
    \textit{Pending Withdrawals:} $\textsf{PW} \gets \emptyset$\\
    \textit{Escrowed Rewards:} $\textsf{RW} \gets \emptyset$\\
    \textit{Epochs:} \textsf{epoch} $\gets \{\, \textsf{chain} \mapsto 0 \mid \ \textsf{chain} \in B \,\}$\\
    \textit{Crossroads addrs:} $\textsf{chain}_\textsf{addr} \gets \{\, \textsf{chain} \mapsto a_\textsf{chain} \mid \textsf{chain} \in B \,\}$\\

    $\circ\:$ On input $(\textsf{``Init''}, \textsf{reward})$ \\
    \hspace*{2em} set $r = \textsf{reward}$\\

    $\circ\:$ On input $(\textsf{``Deposit''}, \textsf{chain}, \textsf{tx})$ \\
    \hspace*{2em} \textbf{assert} $D[\textsf{tx}] = \bot$\\
    \hspace*{2em} $(\textsf{amt}, \textsf{sdr}, \textsf{rsv}, \textsf{fee}) \gets \mathcal{O}(\textsf{chain}, \textsf{tx})$\\
    \hspace*{2em} $\textsf{acct} \gets \textsf{TB}.\textsf{Extract}(\textsf{chain}, \textsf{tx})$\\
    \hspace*{2em} \textbf{assert} $\textsf{amt,acct} \ne \bot \textbf{ and } \textsf{rsv} = \textsf{chain}_\textsf{addr}[\textsf{chain}]$\\
    \hspace*{2em} set $D[\textsf{tx}] = \textsf{Claimed}$\\
    \hspace*{2em} set $\textsf{bal}[(\textsf{acct},\textsf{chain})] = \textsf{bal}[(\textsf{acct},\textsf{chain})] + \textsf{amt}$\\

    $\circ\:$ On input $(\textsf{``Transfer''}, \textsf{chain}, \textsf{amt}, \textsf{rcv})$ from $\textsf{acct}$\\
    \hspace*{2em} \textbf{assert} $\textsf{bal}[(\textsf{acct}, \textsf{chain})] \ge \textsf{amt}$ \textbf{ and } $\textsf{rcv} \in \mathcal{U}$\\
    \hspace*{2em} set $\textsf{bal}[(\textsf{acct}, \textsf{chain})] = \textsf{bal}[(\textsf{acct}, \textsf{chain})] - \textsf{amt}$\\
    \hspace*{2em} set $\textsf{bal}[(\textsf{rcv}, \textsf{chain})] = \textsf{bal}[(\textsf{rcv}, \textsf{chain})] + \textsf{amt}$\\

    $\circ\:$ On input $(\textsf{``LockWithdrawal''}, \textsf{chain}, \textsf{amt})$ from $\textsf{acct}$ \\
    \hspace*{2em} \textbf{assert} $\textsf{bal}[(\textsf{acct},\textsf{chain})] \ge \textsf{amt}$ \textbf{ and } $\textsf{bal}[(\textsf{acct},\eta)] \ge r$\\
    \hspace*{2em} set $\textsf{bal}[(\textsf{acct},\textsf{chain})] = \textsf{bal}[(\textsf{acct},\textsf{chain})] - \textsf{amt}$\\
    \hspace*{2em} set $\textsf{PW}[\textsf{acct}, \textsf{chain}] = \textsf{PW}[\textsf{acct}, \textsf{chain}] + \textsf{amt}$\\
    \hspace*{2em} set $\textsf{bal}[(\textsf{acct},\eta)] = \textsf{bal}[(\textsf{acct},\eta)] - r$\\
    \hspace*{2em} set $\textsf{RW}[\textsf{acct}] = \textsf{RW}[\textsf{acct}] + r$\\
    \hspace*{2em} \textsf{return} $\textsf{id}$ to $\textsf{acct}$\\

    $\circ\;$ On input $(\textsf{``CanSign''}, \textsf{acct}, \textsf{chain}, \textsf{amt}, \textsf{fee})$ from $\mathcal{F}_{\textsf{SC}}$\\
    \hspace*{2em} \textbf{assert} $\textsf{PW}[\textsf{acct}, \textsf{chain}] \geq \textsf{amt} + \textsf{fee}$ \textbf{ and } $\textsf{RW}[\textsf{acct}] \geq r$\\
    \hspace*{2em} \textsf{return} $1$ to $\mathcal{F}_{\textsf{SC}}$\\

    $\circ\;$ On input $(\textsf{``GetEpoch''}, \textsf{chain})$ from $\mathcal{F}_{\textsf{SC}}$\\
    \hspace*{2em} \textsf{return} $\textsf{epoch}[\textsf{chain}]$\\

    $\circ\:$ On input $(\textsf{``ConfirmWithdrawal''}, \textsf{chain}, \textsf{tx})$ from $\textsf{acct}'$ \\
    \hspace*{2em} $(\textsf{amt}, \textsf{sdr}, \textsf{rsv}, \textsf{fee}) \gets \mathcal{O}(\textsf{chain}, \textsf{tx})$\\
    \hspace*{2em} $\textsf{acct} \gets \textsf{TB}.\textsf{Extract}(\textsf{chain}, \textsf{tx})$\\
    \hspace*{2em} \textbf{assert} $\textsf{amt, acct} \ne \bot \textbf{ and } \textsf{sdr} = \textsf{chain}_\textsf{addr}[\textsf{chain}]$\\
    \hspace*{2em} set $\textsf{PW}[(\textsf{acct},\textsf{chain})] = \textsf{PW}[(\textsf{acct},\textsf{chain})] - (\textsf{amt} + \textsf{fee})$\\
    \hspace*{2em} set $\textsf{RW}[\textsf{acct}] = \textsf{RW}[\textsf{acct}] - r$\\
    \hspace*{2em} set $\textsf{bal}[(\textsf{acct}',\eta)] = \textsf{bal}[(\textsf{acct}',\eta)] + r$\\
    \hspace*{2em} set $\textsf{epoch}[\textsf{chain}] = \textsf{epoch}[\textsf{chain}] + 1$\\
}
\caption{The Crossroads assets contract, parameterized by an oracle $\mathcal{O}$ and a transaction binding $\textsf{TB}$ exposing $\textsf{TB}.\textsf{Embed}$ and $\textsf{TB}.\textsf{Extract}$ (see Definition~\ref{def:tx-binding}). $\eta$ denotes the native asset of the backend blockchain. $\mathcal{O}(b, \textsf{tx})$ returns the transferred amount, sender, recipient, and on-chain fee, or $\bot$ if the transaction is not finalized on blockchain $b$.}
\label{fig:crossroads-program}
\end{figure}

\begin{figure}[!t]
\protbox{Signing Committee Protocol $\Pi_{\textsf{SC}}^{\textsf{TSS,TB}}$}
{
    \textit{Parameters:} $n$ signers $\mathcal{S} = \{s_1, \ldots, s_n\}$, 
    threshold $t \leq n$, transaction binding $\textsf{TB}$\\
    \textit{State:} $\mathcal{K} \gets \emptyset$ (per-signer local state)\\
    
    $\circ\;$ On input $(\textsf{``KeyGen''}, \textsf{chain})$\\
    \hspace*{2em} \textbf{assert} $\mathcal{K}[\textsf{chain}] = \bot$ at all $s_i$\\
    \hspace*{2em} $\mathcal{S}$ jointly executes $\Pi_{\textsf{DKG}}(1^\lambda, n, t)$; each $s_i$ obtains $\textsf{out}^{(i)}_{\textsf{chain}}$\\
    \hspace*{2em} parse $\textsf{out}^{(i)}_{\textsf{chain}}$ as $(\textsf{pk}_{\textsf{chain}}, \textsf{pk}_{c,\textsf{chain}}, \textsf{sk}^{(i)}_{\textsf{chain}})$\\
    \hspace*{2em} At each $s_i$: set $\mathcal{K}[\textsf{chain}] = (\textsf{pk}_{\textsf{chain}}, \textsf{pk}_{c,\textsf{chain}}, \textsf{sk}^{(i)}_{\textsf{chain}})$\\
    \hspace*{2em} \textsf{broadcast} $(\textsf{chain}, \textsf{pk}_{\textsf{chain}})$
    \medskip
    
    $\circ\;$ On input $(\textsf{``Sign''}, \textsf{chain}, \textsf{amt}, \textsf{fee}, \textsf{addr}_{\textsf{dest}})$ from $\textsf{acct}$, at each $s_i \in \mathcal{S}$\\
    \hspace*{2em} \textbf{assert} $\mathcal{K}[\textsf{chain}] \neq \bot$\\
    \hspace*{2em} \textbf{assert} $\textsf{crsrds}.\textsf{CanSign}(\textsf{acct}, \textsf{chain}, \textsf{amt}, \textsf{fee}) = 1$\\
    \hspace*{2em} $\textsf{epoch} \gets \textsf{crsrds}.\textsf{GetEpoch}(\textsf{chain})$\\
    \hspace*{2em} $\textsf{tx}_0 \gets \textsf{BuildTx}(\textsf{chain}, \textsf{addr}_{\textsf{dest}}, \textsf{amt}, \textsf{fee}, \textsf{epoch})$\\
    \hspace*{2em} $\textsf{tx} \gets \textsf{TB}.\textsf{Embed}(\textsf{chain}, \textsf{tx}_0, \textsf{acct})$\\
    \hspace*{2em} $\sigma^{(i)} \gets \textsf{TSS}.\mathsf{Sign}(\textsf{sk}^{(i)}_{\textsf{chain}}, \textsf{tx})$\\
    \hspace*{2em} \textsf{return} $(\textsf{tx}, i, \sigma^{(i)})$ to $\textsf{acct}$
    \medskip
    
    $\circ\;$ At $\textsf{acct}$, upon collecting any $t$ shares $\{\sigma^{(j)}\}_{j \in J}$ for $\textsf{tx}$:\\
    \hspace*{2em} $\sigma \gets \textsf{TSS}.\mathsf{Combine}(\textsf{pk}_{c,\textsf{chain}}, \textsf{tx}, J, \{\sigma^{(j)}\}_{j \in J})$\\
    \hspace*{2em} \textbf{if} $\sigma = \mathsf{blame}(J^*)$: exclude $J^*$, collect more shares, retry\\
    \hspace*{2em} \textbf{assert} $\textsf{TSS}.\mathsf{Verify}(\textsf{pk}_{\textsf{chain}}, \textsf{tx}, \sigma) = 1$\\
    \hspace*{2em} output $(\textsf{tx}, \sigma)$
}
\caption{The signing committee protocol $\Pi_{\textsf{SC}}^{\textsf{TSS}}$, parameterized by a $(t, n)$-threshold signature scheme $\textsf{TSS}$ (see~\Cref{sec:tss}) and a transaction binding $\textsf{TB}$ (see~\Cref{def:tx-binding}). Key generation is performed per chain by the committee jointly executing a distributed key generation protocol $\Pi_{\textsf{DKG}}$. On each signing request, committee members independently verify authorization via $\textsf{crsrds}.\textsf{CanSign}$, embed the requesting $\textsf{acct}$ into the withdrawal transaction via $\textsf{TB}.\textsf{Embed}$, and produce signature shares that $\textsf{acct}$ combines into a final signature. The $\mathsf{blame}$ mechanism identifies invalid shares, preserving liveness as long as $t$ honest signers remain.}
\label{fig:signing-committee}
\end{figure}

\subsection{Asset chain requirements}\label{sec:requirements}
Crossroads is compatible with nearly any blockchain, but there are several requirements for a particular integration to take full advantage of our design. All the chains that we have explored meet these requirements, though historically some chains (e.g., Bitcoin) previously have not.

\paragraph{Fee replacements.} Transactions whose fee payments are too low to be included on the target chain must be replaceable by transactions carrying higher fees. This is to prevent a denial-of-service attack in which low-fee transactions are broadcast but never get included in a block.

\paragraph{Single-choice transaction inclusion.} The chain must allow a set of valid, signed transactions containing every potential valid spend from the same account to remain valid until a transaction in the set is selected by the chain's consensus algorithm, after which all the others become invalid forever. On Ethereum, the transaction nonce fills this purpose; on Bitcoin, it can be a UTXO spent in the transaction.

\paragraph{Measurable finalization.} There must be a method for measuring confidence in the finality of a particular transaction using an oracle contract on the backend chain.

\paragraph{Transaction signing support.} The signing committee must support the signature scheme of the chain. Since the committee holds each key under a threshold scheme and never reconstructs it, the scheme must admit a threshold signing protocol whose outputs are accepted by the chain's native verifier.

\paragraph{Transaction binding.} The chain must support a mechanism for attributing transactions in the asset chains to backend blockchain addresses. For deposits, this attribution determines which Crossroads account is credited; for withdrawals, it identifies the spender so that accounting can be applied correctly. We capture this attribution mechanism as an interface implemented per asset chain.
\begin{definition}[Transaction binding]
\label{def:tx-binding}
Let $\mathcal{T}_B$ be the well-formed transactions of blockchain $B$ and $\mathcal{U}$ the backend addresses.
A \emph{transaction binding} for $B$ is a pair of algorithms $\textsf{Embed}: \mathcal{T}_B \times \mathcal{U} \to \mathcal{T}_B$ and $\textsf{Extract}: \mathcal{T}_B \to \mathcal{U} \cup \{\bot\}$ satisfying $\textsf{Extract}(\textsf{Embed}(t,a)) = a$ for all $t \in \mathcal{T}_B$, $a \in \mathcal{U}$.
\end{definition}

Concrete bindings are specific to their respective chains. For example, on Bitcoin, $\textsf{Embed}$ might add an \texttt{OP\_RETURN} output containing $a$. On Ethereum, $\textsf{Embed}$ can append $a$ to the transaction calldata or the access list. We describe these instantiations in Section~\ref{sec:implementation}.

\subsection{Deposits}
\label{subsec:architecture-deposits}
For users to deposit assets on an asset chain into Crossroads, they must transfer them to a Crossroads-controlled encumbered account on the asset chain. Each deposit transaction into Crossroads must include a transaction binding to the account to which the deposit is credited.

After a deposit transaction has been finalized on its asset chain, its sender provides a proof of transaction inclusion to Crossroads. Crossroads validates the proof using the chain-specific oracle, ensures the transaction beneficiary is an account Crossroads controls, and ensures the deposit had not previously been processed. If the checks pass, Crossroads mints a wrapped representation of the deposited tokens to the Crossroads account bound to the deposit and marks the transaction as processed. %

\subsection{Withdrawals}
\label{subsec:architecture-withdrawals}
To withdraw the underlying asset from Crossroads, a user must first burn an amount of tokens needed to cover the withdrawal on the backend blockchain; this establishes the maximum asset-chain transaction cost, which includes both the value being transferred and the transaction fee. The user may then obtain signed transactions from the signing committee spending up to this maximum. Since all withdrawals originate from the same encumbered address, included transactions must be attributable back to the spender. This is achieved by requiring the chain-specific transaction binding in withdrawal transactions, which embeds the spender's backend account into the signed transaction.

It is possible to use, rather than transaction bindings, simply a commitment (on the backend blockchain) from the spender to pay for a particular transaction hash. However, we propose transaction bindings because they do not require the spender to make additional transactions (fresh commitments) if a different party's withdrawal is included instead. In this way, transaction bindings are more resistant to frontrunning and denial of service attacks, since the time when a user requests a signature gets revealed only to the signing committee rather than to the public backend chain.

The spender may choose the transaction fee and hold onto signed transactions indefinitely. Crossroads does not guarantee that a particular transaction it has signed will be included on an asset chain, but it can guarantee that the transactions it signs are valid transactions. Multiple spenders may, in fact, compete to broadcast their transactions from the same asset chain account. After a transaction does land on an asset blockchain and finalizes, a transaction inclusion proof must be provided to Crossroads for accounting purposes. Anyone may submit the proof to receive a small reward. The spender must reserve some tokens in advance to cover this fee.

\section{Optimizations}\label{sec:practical}

The Crossroads construction in~\Cref{sec:crossroads_arch} is functionally complete, but several practical concerns may limit its adoption. On the deposit side, embedding transaction bindings in transaction fields restricts which wallets and applications can produce valid deposits, and the wait for source-chain finality delays deposit times. On the withdrawal side, all withdrawals in the base Crossroads construction from a given chain originate from a single encumbered address, capping throughput at one withdrawal per epoch.

In this section we describe Crossroads optimizations that mitigate these concerns by means of two ideas: use of multiple encumbered addresses on a given asset chain and facilitating fast transactions using third parties. The optimizations we propose can be deployed selectively and do not require modification of the core Crossroads protocol.

\subsection{Deposits}
\label{subsec:optimized-deposits}
We anticipate that Crossroads depositors will not always control their spending key to the extent that they can modify the fields needed to embed a binding in their deposit transaction. Some Crossroads deposits will originate from a centralized exchange or from a wallet which only permits vanilla transfers without embedded transaction data.

We address this with \textit{custom deposit addresses}, each bound to a single Crossroads account to which deposits are credited (the \textit{deposit beneficiary}). These require two asset-chain transactions to complete a deposit: one sending assets to the custom deposit address, and another consolidating those funds into the main encumbered address.

To do so, we employ a novel auction mechanism selling a sub-balance of the custom deposit address representing the new deposit in exchange for equivalent Crossroads tokens. The auction's reserve price is equal to the deposited amount minus the cost the deposit beneficiary would have paid in transaction fees to send the consolidation transaction herself.

This auction is enabled by key encumbrance. The auction winner gains encumbered control over the custom deposit addresses it bid on, allowing it to consolidate them at any later time. Bidders are incentivized to participate: the winner can delay broadcasting the consolidation transaction until transaction fees fall, profiting from the difference. In effect, the auction winner assumes the role of performing the same type of asset consolidation optimization that centralized exchanges perform~\cite{coinbasebatching}.

\mypara{Speeding up deposits.}
Before issuing tokens, Crossroads needs to wait until a deposit transaction is finalized on the asset chain in order to protect against multiple claims to the same deposit (i.e., a double spend). On some chains this finalization is slow, taking up to 30 minutes or longer~\cite{circle-cctp-finality}.

We can improve deposit times by transferring the risk of double spending from Crossroads to a third party willing to take it for a fee. A small extension to the asset contract lets a depositor reassign an unfinalized deposit, enabling a separate fast-deposit contract to atomically swap a pending deposit claim for already-issued Crossroads tokens owned by the third party. Suppose, for example, that Alice deposits 10 ETH to Crossroads but does not want to wait for finality (approximately 15 minutes). Bob already has Crossroads ETH tokens and is willing to accept Alice's deposit in exchange for his tokens, for a small fee, after it gets one block confirmation (under 12 seconds). Alice atomically swaps her deposit claim for 9.999 Crossroads ETH tokens from Bob. Once the deposit finalizes, Bob submits the deposit proof and is issued 10 Crossroads ETH tokens for a profit of 0.001 ETH. This mechanism mirrors the fast fills of intent-based bridges such as Across~\cite{across}, where relayers provide quick access to bridged funds for a fee; in Crossroads, however, the third party is repaid through the standard {deposit flow}.

\subsection{Withdrawals}\label{sec:faster-withdrawals}

On most account-based blockchains, deposits to Crossroads can happen in parallel, since each deposit transaction is independent of all other deposits and all withdrawals. In our~\Cref{sec:crossroads_arch} construction, however, withdrawals from the same chain all originate from the same encumbered address, so the system can only execute one withdrawal at a time per chain, after which it must wait for finalization and inclusion proof before signing the next.

We address this bottleneck with an intent-based withdrawal contract on the backend chain. To withdraw, a user locks the desired amount plus a small reward in Crossroads tokens in the contract. Any party, called a \textit{filler}, may then send the equivalent native asset to the user's destination address on the target chain from their own address. The filler then submits an inclusion proof to claim the locked tokens. If no filler claims the intent within a timeout, the user reclaims the locked tokens and falls back to the standard withdrawal path. As the asset transfer happens from the filler's own address, it does not conflict with the encumbered account's transactions, and many withdrawals can be served in parallel.

Without coordination, however, two fillers may race to fulfill the same request, with only one being reimbursed. To prevent this, fillers commit to a request by staking collateral on the backend chain; the contract accepts at most one commitment per intent, so any filler that commits is guaranteed to be the sole eligible claimant. If the committed filler fails to complete the transfer within the timeout, the stake is forfeited to the user as compensation for the delay.

A second route to higher throughput is to use multi-payment contracts on asset chains that support them, executing many transfers in a single transaction. Users wishing to withdraw in the next epoch lock the corresponding tokens and specify their destination addresses; the signing committee then signs a single transaction transferring all locked funds since the last epoch.

Finally, withdrawal throughput can also be increased by having Crossroads control several encumbered addresses per chain rather than one. Since the per-chain bottleneck in~\Cref{sec:crossroads_arch} stems from every withdrawal originating from a single address, distributing the chain's liquidity across multiple encumbered addresses lets the signing committee sign withdrawals from each in parallel, one per address per epoch. The trade-off is that liquidity becomes fragmented across addresses, so withdrawals must be routed to an address holding sufficient funds, and occasional rebalancing is needed when balances drift.

\section{Security of Crossroads}\label{sec:security}

The main security property of Crossroads relates to asset ownership. We define the system as sound if and only if any user can, at any point, unilaterally withdraw their Crossroads net balance to externally owned accounts (EOAs) of their choosing. The net balance is calculated as the sum of all deposited assets, adjusted for the net value of internal transfers, minus any assets that have already been withdrawn. We show that this result implies both correctness and soundness of the system: under the threat model outlined below, no matter the circumstances, a user is always able to withdraw the assets it owns and no adversary can mint or burn tokens deliberately.

Crossroads makes several trust assumptions to achieve soundness, beyond the threshold assumption on the signing committee discussed in~\Cref{sec:security}. We articulate them here.

\subsection{Threat model}
\label{subsec:threat-model}

\mypara{Backend blockchain.} We assume the backend blockchain provides standard execution guarantees: transactions are executed correctly, finalized transactions are permanent, and the chain remains available. A compromise of the backend blockchain (for instance, through a successful attack on its consensus) could lead to corruption of contract state, allowing an adversary to double spend withdrawals. We therefore recommend deploying Crossroads on a backend blockchain whose security guarantees match or exceed those of the asset chains it serves.

\mypara{Signing committee.} Crossroads relies on a $(t,n)$-threshold signing committee. We assume an adversary corrupts strictly fewer than $t$ committee members. Beyond this baseline, we present additional mechanisms (\Cref{sec:security_enhacements}) to mitigate the impact of committee compromise. To achieve soundness we also require the signing committee to be \emph{live} and acting strictly under $\Pi^{\mathsf{TSS}}_{\mathsf{SC}}$.

\mypara{Oracles.} Crossroads relies on per-chain oracles to relay finalized transaction information from asset chains to the backend blockchain. We assume each oracle  only confirms transactions that have been finalized on the source chain. We additionally assume oracles are \emph{live}, meaning that any party can obtain a confirmation for any finalized transaction at any point. A compromised oracle could falsely confirm deposits that never occurred, causing the asset contract to mint unbacked wrapped tokens, or fail to confirm withdrawals, freezing user funds. Different oracle constructions satisfy oracle soundness under different trust assumptions: zkBridge-based oracles rely on the soundness of zero-knowledge proofs and the source chain's consensus, while TEE-based oracles trust the integrity of the hardware enclave. This makes this trust assumption \emph{application specific}. However, the impact of an oracle compromise is contained to its own integration: each asset contract mints tokens only against attestations from its designated oracle, so a faulty oracle for one chain can inflate or freeze only the wrapped assets of that chain.

\subsection{Crossroads soundness}

Let $\mathcal{U}$ be the set of registered accounts to the Crossroads system, and $\mathcal{B}$ the set of asset blockchains which fulfill the requirements of~\ref{sec:requirements}. We model the system's evolution as an ordered sequence of \emph{operations}, where each operation is either a call to the Crossroads asset contract  $\textsf{crsrds}(\mathcal{O}, \mathcal{T}_\mathcal{B})$, a message in the signing-committee protocol $\Pi^{\mathsf{TSS}}_{\mathsf{SC}}$, or the broadcast of a transaction to an asset chain. 
A \emph{transcript} $\vec{\sigma}$ is a finite, ordered sequence of operations. We say an operation is \emph{$\mathsf{acct}$-issuable} if it can be performed using only $\mathsf{acct}$'s backend signing key; in particular, it requires no cooperation from any other backend account, oracle, or external party.

We assume each asset chain $\mathsf{chain}$ defines a validity predicate $\mathsf{Valid}_{\mathsf{chain}}(\mathsf{tx}, \mathsf{st})$ on transactions and chain states, capturing the chain's inclusion checks: well-formedness, signature verification under the sender's key, epoch validity (or non-conflict with prior spends), and assurance the sender's native balance covers $\mathsf{value}(\mathsf{tx}) + \mathsf{fee}(\mathsf{tx})$.

Fix an account $\mathsf{acct} \in \mathcal{U}$, a blockchain $\mathsf{chain} \in \mathcal{B}$, and a transcript $\vec{\sigma}$, and consider the operations of $\vec{\sigma}$ that call $\mathsf{chain}$'s asset contract $\textsf{crsrds}(\mathcal{O}, \mathcal{T}_\mathcal{B})$ and execute successfully, i.e., that pass the contract's assertions and so update balances. Among these, let $\mathsf{Dep}(\mathsf{acct})$ be the \textsf{Deposit} calls credited to $\mathsf{acct}$ and $\mathsf{With}(\mathsf{acct})$ the \textsf{ConfirmWithdrawal} calls attributed to $\mathsf{acct}$ via the transaction binding, and let $\mathsf{Rcv}(\mathsf{acct})$ and $\mathsf{Snd}(\mathsf{acct})$ be the \textsf{Transfer} calls in which $\mathsf{acct}$ is the recipient and the sender, respectively. For each such operation, $\mathsf{amt}$ denotes the amount deposited, withdrawn, or transferred; for $\mathsf{tx} \in \mathsf{With}(\mathsf{acct})$, $\mathsf{fee}$ additionally denotes the on-chain withdrawal fee reported by $\mathcal{O}$. We define the \emph{net balance} of $\mathsf{acct}$ for $\mathsf{chain}$ with respect to $\vec{\sigma}$ as:\begin{align*}
\mathsf{Net}(\mathsf{acct}, \mathsf{chain}, \vec{\sigma}) \;=\;
\sum_{\substack{\mathsf{tx} \in \vec{\sigma} \\ \mathsf{Dep}(\mathsf{acct})}} \mathsf{amt}(\mathsf{tx}) \;+\; \sum_{\substack{\mathsf{tx} \in \vec{\sigma} \\ \mathsf{Rcv}(\mathsf{acct})}} \mathsf{amt}(\mathsf{tx}) \\
-\; \sum_{\substack{\mathsf{tx} \in \vec{\sigma} \\ \mathsf{With}(\mathsf{acct})}} \big(\mathsf{amt}(\mathsf{tx}) + \mathsf{fee}(\mathsf{tx})\big) \;-\; \sum_{\substack{\mathsf{tx} \in \vec{\sigma} \\ \mathsf{Snd}(\mathsf{acct})}} \mathsf{amt}(\mathsf{tx}).
\end{align*}

\begin{definition}[Soundness]
\label{def:soundness}
We say Crossroads is \textbf{sound} if, for any transcript $\vec{\sigma}$, $\mathsf{chain} \in \mathcal{B}$, $\mathsf{acct} \in \mathcal{U}$, $\mathsf{addr}_{\mathsf{dest}}$, $\mathsf{amt}$, $\mathsf{fee}$ with $\mathsf{Net}(\mathsf{acct}, \mathsf{chain}, \vec{\sigma}) \ge \mathsf{amt} + \mathsf{fee}$, and with $\mathsf{bal}[(\mathsf{acct}, \eta)] \ge r$ in the contract state determined by the execution of $\vec{\sigma}$, $\mathsf{acct}$ can produce a finite ordered sequence $S$ of $\mathsf{acct}$-issuable operations such that, for every PPT adversary, with all but negligible probability over the coins of $\Pi_{\textsf{DKG}}$, $\textsf{TSS}$, and the adversary, every transcript $\vec{\sigma}'$ extending $\vec{\sigma}$ that contains $S$ in order includes a transaction $\mathsf{tx}$ broadcast to $\mathsf{chain}$ with $\mathsf{recipient}(\mathsf{tx}) = \mathsf{addr}_{\mathsf{dest}}$, $\mathsf{value}(\mathsf{tx}) = \mathsf{amt}$, and $\mathsf{Valid}_{\mathsf{chain}}(\mathsf{tx}, \mathsf{st}_{\mathsf{chain}}) = 1$, where $\mathsf{st}_{\mathsf{chain}}$ is the chain state at the moment $\mathsf{tx}$ is broadcast.
\end{definition}

Intuitively, soundness captures the end-to-end guarantee that no sequence of valid operations can deprive a user of their assets: whatever value a user has deposited into Crossroads, adjusted for internal transfers and prior withdrawals, remains recoverable.

Soundness thus guarantees a \emph{valid broadcast}, not inclusion. All withdrawals spend from the same encumbered address, so signed transactions for the same epoch are mutually exclusive, and a competing withdrawal may be included first. Losing such a race costs latency, never funds. By single-choice inclusion (Section~3.2), the superseded transaction becomes permanently invalid, while the user's pending-withdrawal balance is untouched, so the user may request a fresh signature for the next epoch and rebroadcast.

Moreover, soundness is a property of the Crossroads protocol, not of the applications built on it. It guarantees that whoever holds a net balance can withdraw it, not that a faulty application won't take that balance from you. If a flawed AMM or lending contract drains your tokens, soundness still holds (the new holder can withdraw), but your assets are gone. As on any smart-contract platform, application-layer security is orthogonal to soundness of the underlying system.

We now prove that the Crossroads functionalities satisfy this property under the threat model presented in \ref{subsec:threat-model}.

\begin{theorem}[System Soundness]
\label{thm:soundness}
Let $\mathsf{chain} \in \mathcal{B}$ be a compatible blockchain satisfying Oracle Soundness, and instantiate the signing committee protocol $\Pi_{\textsf{SC}}^{\textsf{TSS}}$ with distributed key generation and a $(t, n)$-threshold signature scheme $\textsf{TSS}$ that is EUF-CMA secure against an adversary corrupting up to $t-1$ signers. Then, under the Crossroads asset contract $\textsf{crsrds}(\mathcal{O}, \mathcal{T}_{\mathcal{B}})$ and $\Pi_{\textsf{SC}}^{\textsf{TSS}}$, Crossroads satisfies soundness (Definition~\ref{def:soundness}) against any PPT, resource-bounded adversary. 
\end{theorem}

The proof is deferred to Appendix~\ref{appdx:proofs}. The assumption on the bound of adversarial resources limits the number of withdrawals the adversary can finalize to compete with an honest user's, and hence the number of times it can deny any single withdrawal; we discuss denial-of-service attacks, along with mitigations, in Appendix~\ref{subsec:dos}.

\subsection{Modular security enhancements}\label{sec:security_enhacements}

Committee-based protocols that custody large volumes of funds are frequent attack targets~\cite{krause20251,chainalysis2026kelpdao,zhang2024cross}, so we now present several mechanisms that strengthen the security guarantees of Crossroads beyond the threshold assumption.

\paragraph{Heterogeneous TEE committee.} A first mechanism is to require (and enforce via attestation) that all committee nodes run within TEEs. Using a diverse mix of hardware architectures---such as Intel SGX, AMD SEV, and AWS Nitro Enclaves---further ensures that a hardware-level vulnerability specific to one manufacturer cannot compromise all committee members at once.

\mypara{Security Council.} Another layer of security would be to have an emergency council with the power to freeze Crossroads contracts. This council could act as a decentralized governance layer designed to intervene during technical failures that automated systems might not catch. While the signing committee provides cryptographic security, the Council provides oversight to protect user capital from unforeseen exploits.

\mypara{Hybridization.} When Crossroads is integrated with a blockchain that has smart contract capabilities, it is possible to introduce an additional layer of security that remains effective even in the event of a TEE and committee compromise. In this scenario, funds can be stored within a smart contract on the native blockchain. This allows the smart contract to enforce programmatic security measures, such as limiting withdrawal throughput to prevent mass drainage, or providing a freeze functionality that can be activated by an external security council via a threshold signature.

\mypara{Accountable threshold signing.} Standard threshold signatures conceal which committee members contributed to a given signature, limiting the system's ability to hold misbehaving signers accountable. Accountable-subgroup multisignatures~\cite{10.1145/501983.502017} instead produce signatures that cryptographically reveal the signing subgroup to any verifier, with no reliance on off-chain transcripts or trusted logs. Adopting such a scheme in Crossroads would allow any signature produced outside the authorized withdrawal flow to be traced back to the participating signers, enabling on-chain slashing or removal from the committee.

\section{Implementation}\label{sec:implementation}

We provide a complete prototype of Crossroads with three blockchains currently integrated: Bitcoin, Ethereum, and Solana. Others may be integrated permissionlessly. The prototype uses a local EVM devnet as the backend blockchain as well as a simple MPC signing committee consisting of three signers. We make the source code of our prototype publicly available at {\url{https://github.com/trate3/crossroads}}. 

\subsection{Setup}
We instantiate the signing committee using the \textit{threshold-signatures} library by NEAR~\cite{near_threshold_signatures_2026}.
Committee members participate in threshold signing under MPC to compute signatures---specifically, threshold ECDSA using a two-phase signing protocol~\cite{damgrard20fastthreshold} and FROST-Ed25519~\cite{komlo2020frost}. ECDSA and Ed25519 cover the majority of signing algorithms used in current blockchains. 

\paragraph{Bootstrap phase.} The signing committee requires a bootstrap phase for establishing the initial membership of the committee as well as generating the committee's MPC root key. During the bootstrap phase, a smart contract is posted on the backend blockchain. Each member of the signing committee registers its identity to the smart contract. In our prototype, the committee members must be approved by a coordinator of the bootstrap contract (who has no other authority in the system) but in practice this step could require additional certification, such as a remote attestation of the committee members' node software.

Once all the committee members are registered, the bootstrap coordinator distributes connection information (e.g. network addresses) to the committee members, who then begin distributed key generation and create the root secret key material.

\paragraph{Asset contracts.} We cleanly separate the Crossroads asset contract into two components: a chain-agnostic ERC-20 contract which is reused across all Crossroads tokens, and a contract responsible for the chain-specific implementation. The chain-specific contract must decode deposit/withdrawal transactions and transaction bindings, prove transaction inclusion (often with a separate oracle), and decide how much a serialized transaction costs or spends on its asset chain.

\subsection{Ethereum}
Ethereum's account-based address system is well suited for Crossroads; deposits into an encumbered Ethereum account never not impose additional costs on withdrawals.

\paragraph{Epochs.} Ethereum transactions contain a natural gating mechanism, the nonce, as all transactions from the same signer must be included in nonce order. When a withdrawal transaction is proven included, the asset contract then allows signatures using the next nonce. For simplicity, only typical type-2 transactions may be signed, which prevents newer transaction types from interfering with our accounting.

\paragraph{Transaction binding.} As described in~\Cref{subsec:architecture-withdrawals}, deposits and signed withdrawal transactions must be attributable to a Crossroads account, and this is accomplished through a transaction binding. Although Ethereum transactions have a data field, this field is used for smart contract interactions.
Fortunately, extra data appended after the required transaction data is often ignored by smart contracts. Therefore, we use the final 20 bytes of the transaction calldata as the transaction binding.

\paragraph{Oracle.} Ethereum is secured through proof of stake, which cannot be verified succinctly in smart contracts, so we rely on an external block hash oracle. Transactions and their receipts are checked through Merkle proofs to their respective roots in the block header.

\subsection{Bitcoin}
Bitcoin follows the unspent transaction outputs (UTXO) model: transactions consist of at least one input (a previous unspent output) and at least one output. This affects our deposit model. If Crossroads were to accept payments from depositors as a typical Bitcoin wallet does, it would accumulate a set of spendable UTXOs. Bitcoin transaction fees are charged by the size of a transaction, so spending a large number of small UTXOs, each of which must be referenced in the transaction, adversely affects withdrawal costs.

We remedy this by requiring depositors to perform the UTXO consolidation (similar to the custom address deposits of~\Cref{subsec:optimized-deposits}). Rather than collecting new UTXOs, the Bitcoin asset contract instead maintains a single UTXO per encumbered account. The asset contract permits anyone to request transaction signatures with the current encumbered UTXO as an input---but only if the requester combines it with another input resulting in a net deposit to the encumbered account. When inclusion is proven for a Bitcoin transaction spending the encumbered account's UTXO, the asset contract updates the current UTXO to the newly created one and increments the epoch number.

\paragraph{Epochs.} The asset contract's single UTXO is initialized using a deposit transaction with the transaction binding, but all future deposits must use the encumbered account's last UTXO. UTXOs from deposits made to the encumbered account which do not conform to the binding can be handled as separate deposit accounts as described in~\Cref{sec:practical}, but this is not part of our prototype.

\paragraph{Transaction binding.} The transaction binding we choose for Bitcoin transactions is a zero-value \texttt{OP\_RETURN} output in the transaction, which contains a Crossroads address. For deposits, the address is where deposited funds should be issued, and for withdrawals, it is the spender account (enforced by the smart contract).

\paragraph{Oracle.} Our transaction oracle for Bitcoin validates a Merkle proof of transaction inclusion within the Bitcoin block in which it was mined. The block header can be verified using Bitcoin's simplified payment verification scheme, i.e., a light client of Bitcoin which runs directly on the backend blockchain.

The requirement for signatures from the signing committee during the consolidation step of a deposit, spending the same Crossroads UTXO a withdrawal would spend, introduces an additional bottleneck not present in deposits on account-based blockchains.
We note that combining multiple deposits or withdrawals into a single transaction would further increase throughput and reduce fees. Furthermore, multiple encumbered accounts can be used in parallel to account for this limitation (similar to how centralized exchanges maintain multiple hot wallets on each blockchain).

\subsection{Solana}
Like Ethereum, Solana is an account-based blockchain and accepts deposits into a Crossroads encumbered account without requiring interaction with the signing committee. Solana transactions are composed of a series of instructions to \textit{programs} (equivalent to smart contracts).

\paragraph{Epochs.} Solana does not have a native per-account nonce by default, but it does support one, called the ``durable nonce''~\cite{solana_durable_nonces}. A one-time payment initializing a durable nonce account must be made to cover the cost of storing the nonce. This requires an initialization step of two System Program instructions calling \texttt{CreateAccount} and \texttt{InitializeNonceAccount} on a fresh public key derived from the committee, separate from the primary account which holds deposited SOL tokens. The authority of this nonce account is set to the primary account during this step.\footnote{In Solana, the nonce account's ``authority'' is the public key authorized to advance the nonce.}

\paragraph{Transaction binding.} All signed Crossroads withdrawals contain three instructions: (1) \texttt{AdvanceNonceAccount}, (2) \texttt{Transfer}, and (3) \texttt{Memo}, which includes the transaction binding (the spender's Crossroads account). If a signed transaction with an incorrect nonce is sent to the network (e.g., when a different transaction containing the same nonce has already been included), the transaction is considered invalid and is dropped without penalty.

\paragraph{Oracle.} Solana does not have a transaction Merkle tree, so we instead integrate a TEE-based transaction oracle service. This service aggregates the responses of a quorum of independent Solana RPC providers to the finalization state of a given transaction and, if they agree, produces a response signed by a key that itself is bound to the program's attestation. The oracle smart contract verifies the attestation and adds the transaction hash to a set of finalized transactions.

\subsection{Evaluation}
We evaluated our implementation to assess the feasibility of running our system.

Although we used a free local Ethereum devnet as the backend blockchain for most of our testing, in practice Crossroads would be deployed on a public EVM blockchain.
In~\Cref{tab:impl-transaction-costs} we report the cost of using Crossroads to deposit, transfer, and withdraw each of the three assets we integrated on five compatible backend blockchains. Unsurprisingly, L2 blockchains---Optimism, Base, and Arbitrum---offer better fee rates than Ethereum, as does Oasis Sapphire, a TEE-based blockchain used for privacy applications described in~\Cref{sec:applications}. Moreover, as shown in \Cref{appdx:signing-latency}, Crossroads deposits and withdrawals cost little more than ordinary native transfers on each asset chain.

\paragraph{Throughput.} The base Crossroads design is constrained by withdrawals: there is one withdrawal per epoch per encumbered address, and each epoch is as long as the transaction finalization oracle requires. Under maximum security oracles that wait for six confirmations on Bitcoin and full finality on Ethereum and Solana, this means one transaction per encumbered account every 60 minutes, 15 minutes, and 13 seconds, respectively. Nonetheless, as described in~\Cref{sec:faster-withdrawals}, withdrawal throughput can easily be scaled across multiple encumbered accounts associated with that asset.

\paragraph{Latency.} If an encumbered account does not currently have a withdrawal waiting to finalize, it is immediately available for withdrawal use; otherwise, a user must wait for the existing transaction to finalize. We report the latency of signature requests using our prototype in~\Cref{appdx:signing-latency}.

\paragraph{Signing committee resources.} We measured the resources required to run an individual signing committee member under multiple committee sizes. Even with a committee size of 15, each node used under 20 MB of memory to run the signer application (on top of OS memory) and at 20 signing requests per second used approximately 637.6 KB/s of network bandwidth, split equally in both directions. Thus, a committee node would fit in most small-sized cloud compute instances at low cost.

\begin{table}[t]
	\centering
	\scalebox{0.82}{%
	\begin{tabular}{@{}ll|lllll@{}}
    \toprule
	Action & Asset & Ethereum & Optimism & Base & Arbitrum & Sapphire \\
	\midrule
	\multirow{3}{*}{\makecell[l]{Prove\\deposit}}
	& Ethereum & \$0.1418 & \$0.000573 & \$0.00323 & \$0.0109   & \$0.000205 \\
	& Bitcoin  & \$0.1564 & \$0.000625 & \$0.00355 & \$0.0120   & \$0.000226 \\
	& Solana   & \$0.0733 & \$0.000289 & \$0.00166 & \$0.00564 & \$0.000106 \\
	\midrule
	\makecell[l]{Transfer} & (any) & \$0.0223 & \$0.000090 & \$0.00051 & \$0.00173 & \$0.000032 \\
	\addlinespace[2pt]
	\makecell[l]{Burn for\\spending} & (any) & \$0.0487 & \$0.000189 & \$0.00110 & \$0.00374 & \$0.000071 \\
	\midrule
	\multirow{3}{*}{\makecell[l]{Prove\\with-\\drawal}}
	& Ethereum & \$0.1411 & \$0.000572 & \$0.00321 & \$0.0109   & \$0.000204 \\
	& Bitcoin  & \$0.1584 & \$0.000635 & \$0.00360 & \$0.0122   & \$0.000229 \\
	& Solana   & \$0.0758 & \$0.000301 & \$0.00172 & \$0.00584 & \$0.000110 \\
	\bottomrule
	\end{tabular}%
	}
	\caption{Crossroads operation costs in USD targeting our three asset integrations, under different backend chains.
	The backend native token prices, gas prices, and L1 data fees are based on those from June 10, 2026: ETH \$1{,}624.48, BTC \$61{,}575.66, SOL \$64.03, and ROSE \$0.006203.}
	\label{tab:impl-transaction-costs}
\end{table}

\section{Crossroads Applications}\label{sec:applications}

\newcolumntype{a}{>{\columncolor{white}}c}
\newcommand{\fwcell}[2]{%
	\makecell[t]{%
		\begin{minipage}[t]{\getColumnWidth{#1}}
			\raggedright
			#2
		\end{minipage}%
	}
}

\newcommand{\getColumnWidth}[1]{%
	\ifcase#1
	\or
	\or
	\or 6.5cm %
	\or 6.5cm %
	\or 2.8cm %
	\else 4cm
	\fi
}

\definecolor{carnelian}{HTML}{B31B1B}

\begin{figure*}[!th]
    \renewcommand{\arraystretch}{1}
    \resizebox{\textwidth}{!}{
    \begin{NiceTabular}{a|cccc}[color-inside, cell-space-top-limit=2pt, cell-space-bottom-limit=4pt]
    \CodeBefore
    \rowcolors{2}{gray!20}{}
\Body
     \toprule
     \rowcolor{gray!50}
     \textbf{Category} & \textbf{Application} & \textbf{Addressed Issue} & \textbf{Crossroads Enables} & \textbf{Backend Contract} \\
     \midrule
    
    \multirow{8}{*}{\makecell[t]{Wallets \&\\ Liquidity}} & \multirow{1.6}{*}{\makecell[t]{\textbf{Chain-Agnostic} \\ \textbf{Wallet}}} & \fwcell{3}{Managing assets across blockchains often requires separate wallets and tooling} & \fwcell{4}{A single Crossroads account to send transactions on all integrated chains} & \fwcell{5}{---} \\
    
    & \multirow{3}{*}{\makecell[t]{\textbf{Cross-Chain} \\ \textbf{DEXs}}} & \fwcell{3}{Cross-chain swaps support few networks and charge high fees; users fall back on CEXs or bridge aggregators} & \fwcell{4}{Standard AMM contracts to perform atomic swaps between native assets from arbitrary chains} & \fwcell{5}{AMM pool (e.g., Uniswap~\cite{uniswapv3})} \\
    
    & \multirow{3}{*}{\makecell[t]{\textbf{Universal} \\ \textbf{Receive Addresses}}} & \fwcell{3}{Donation recipients and merchants often accept only particular assets, forcing senders to exchange funds themselves} & \fwcell{4}{Recipient-bound receive addresses on all integrated chains, with automatic exchange to a chosen token} & \fwcell{5}{Custom deposit address and swap contracts} \\

    \midrule

    \multirow{13}{*}{\makecell[t]{Universal \\ Assets}} & \multirow{3}{*}{\makecell[t]{\textbf{Universal} \\ \textbf{Stablecoins}}} & \fwcell{3}{Stablecoin chain support is selected by centralized issuers (e.g., USDT is on only 13 of 150+ active chains)} & \fwcell{4}{Permissionless stablecoin deployment and spending on any integrated chain} & \fwcell{5}{Stablecoin issuance contract} \\
    
    & \multirow{3}{*}{\makecell[t]{\textbf{Cross-Chain} \\ \textbf{Staking}}} & \fwcell{3}{Staking protocols can only be secured by native assets, limiting them to their own chain's liquidity and stability} & \fwcell{4}{Securing protocols on one chain with liquid, stable assets from another (e.g., ETH or BTC)} & \fwcell{5}{Staking / slashing contract} \\

    & \multirow{3}{*}{\makecell[t]{\textbf{Universal} \\ \textbf{Testnet Faucet}}} & \fwcell{3}{Obtaining testnet funds relies on faucets that are frequently rate-limited, depleted, or gated behind existing reputation or identity} & \fwcell{4}{Faucets to convert user-provided value (e.g., via PoW) directly to test network funds, and users to integrate new test networks} & \fwcell{5}{Faucet token contract} \\

    & \multirow{3}{*}{\makecell[t]{\textbf{Cross-Chain} \\ \textbf{Lending}}} & \fwcell{3}{Borrowing against assets on one chain to obtain assets on another requires custodial intermediaries} & \fwcell{4}{Collateral from any chain to borrow any Crossroads asset, enabling atomic cross-chain shorts} & \fwcell{5}{Lending pool (e.g., Morpho~\cite{morpho} or Euler~\cite{euler})} \\

    \midrule

    \multirow{10}{*}{\makecell[t]{Privacy \&\\ Compliance}} & \multirow{2.7}{*}{\makecell[t]{\textbf{Private} \\ \textbf{Payments}$^\ast$}} & \fwcell{3}{On-chain transactions are publicly visible; CEX-level privacy requires trusting a custodian~\cite{cexdex2026report}} & \fwcell{4}{Private transactions which only leak the size of deposits and withdrawals when funds enter or leave Crossroads} & \fwcell{5}{---} \\

    & \multirow{3}{*}{\makecell[t]{\textbf{Private Asset} \\ \textbf{Management}$^\ast$}} & \fwcell{3}{Institutional asset management requires privacy, yet most liquidity lies on public chains, forcing a custody–liquidity trade-off} & \fwcell{4}{Private institutional management of public-chain assets with fine-grained disclosure policies for balances and transfers} & \fwcell{5}{Disclosure-management contract} \\

    & \multirow{3}{*}{\makecell[t]{\textbf{On-Chain} \\ \textbf{Compliance}}} & \fwcell{3}{Compliance typically requires off-chain enforcement by centralized, opaque entities} & \fwcell{4}{Transparent, smart-contract-enforced compliance policies (e.g., whitelisting), composable with private applications} & \fwcell{5}{KYC, whitelisting, and other policy contracts} \\

    \bottomrule
    \CodeAfter
        \tikz \fill [carnelian] (3-|2) rectangle ($(4-|2)!0.07!(4-|3)$) ;
        \tikz \fill [carnelian] (7-|2) rectangle ($(8-|2)!0.07!(8-|3)$) ;
        \tikz \fill [carnelian] (9-|2) rectangle ($(10-|2)!0.07!(10-|3)$) ;
\end{NiceTabular}
}
\caption{Table of example applications Crossroads enables. All applications rely on the ERC-20 asset contracts; the \textbf{Backend Contract} column lists additional smart contracts required for each application. Applications marked with $^\ast$ require a private backend blockchain (e.g., Oasis Sapphire). Rows marked in \textcolor{carnelian}{red} indicate applications we have implemented.}
\label{fig:crossroads-applications}
\end{figure*}

Crossroads enables many different applications that interact with assets across multiple chains. We present a range of these applications, summarized in~\Cref{fig:crossroads-applications}, each falling under one of three headings: (1) Wallets and liquidity, (2) Universal assets, and (3) Privacy and compliance. We have implemented one application in each category, which we describe in this section. We defer the discussion of the others to~\Cref{appdx:other_applications}.

\subsection{Wallets and liquidity}

Crossroads' most immediate application lies in how users hold, send, and exchange tokens once all chains' assets share the same programmable backend. Today, this generally requires a separate wallet and/or gas token per chain. Crossroads' representation of each asset in the backend chain as ERC-20 tokens, however, yields a \emph{chain-agnostic wallet} from a single key pair, enables \emph{universal receive addresses} that auto-convert received funds to a recipient's asset of choice, and enables the use of standard decentralized exchange (DEX) smart contracts to enable fast, cheap, and transparent \emph{cross-chain exchange}.

\mypara{Cross-chain decentralized exchange.} The current cross-chain exchange landscape remains heavily centralized. Some decentralized services offer cross-chain exchange but are typically opaque, permissioned, and not fee-efficient. This leaves many smaller blockchains excluded from major interoperability frameworks and forces users onto inefficient or custodial pathways to move assets between them.

Crossroads provides a transparent, permissionless, and fee-efficient alternative. Since exchanges within Crossroads are equivalent to atomic swaps between the corresponding native assets, any standard DEX contract on the backend blockchain can serve as a cross-chain exchange for the cost of a single backend chain transaction. This includes AMMs such as Uniswap~\cite{uniswapv3}, on-chain order books, RFQ systems, batch auctions, and off-chain order matching systems, such as 0x Protocol~\cite{0xprotocol} or Airswap~\cite{airswap}.

In our implementation, we ported Uniswap V2 to a backend chain and launched liquidity pools between Crossroads assets. If Crossroads is deployed on a backend chain that is already well integrated with major assets, Crossroads tokens can route through those existing pools.

\subsection{Universal assets}

A second class of applications explores how cross-chain assets can enable applications previously constrained to one chain, such as staking, lending or stablecoin issuance.

Of these, we have implemented a universal testnet faucet, which converts proof-of-work from one chain directly into testnet funds from any integrated testnet.

\mypara{Universal testnet faucet.}  Existing testnet token faucets are frequently rate limited, depleted, or gated behind existing reputation or identity. This adds friction for developers to test their applications and makes it difficult for privacy-conscious users to obtain test funds. Moreover, AI coding agents are increasingly used to develop and test applications autonomously, yet they cannot clear the identity proofs that faucets require, forcing a human back into the loop.

We have implemented a faucet in Crossroads that converts user-supplied proof-of-work directly into funds on a target testnet. To obtain testnet tokens, a user or agent directs hashpower toward a proof-of-work chain. In our implementation, this is Monero, which supports CPU-based mining with its use of the RandomX hash function. The resulting coins are deposited into Crossroads and automatically swapped for the desired testnet tokens through a backend DEX. Because the only requirement is compute, our faucet provides funds without rate limits, reputation, identity requirements, or other anti-Sybil protections.

Crossroads' permissionless integration of new chains allows anyone to add new testnets.
Testnet token liquidity is ensured through arbitrage: arbitrageurs are incentivized to refill testnet tokens into the DEX if its reserves run low.

\subsection{Privacy and compliance}

Users of centralized exchanges enjoy privacy that public blockchains lack: their balances and trades are visible only to the exchange. Crossroads aims to match this guarantee in a decentralized setting. By choosing a private backend blockchain, transfer transactions within Crossroads remain fully private and withdrawal details are visible only to its signing committee, enabling \emph{private payments} and \emph{private asset management}. To balance the concerns such privacy raises, compliance policies such as KYC whitelisting can be enforced transparently by smart contracts on the same transactions (\emph{on-chain compliance}).

We have realized such privacy in Crossroads through implementation with  a \emph{private} backend blockchain, namely Oasis Sapphire~\cite{cheng2019ekiden, oasis2023sapphire}.

\mypara{Private payments.} Oasis Sapphire uses TEEs to conceal contract state and calldata, exposing them only to parties the contract itself authorizes. In a Sapphire-based Crossroads deployment, therefore, internal activity is  concealed: where transfers and exchanges never leave the backend chain, the asset contract might disclose an account's balance only to its owner; then, external observers of the backend chain would learn only metadata (a transaction's timing, destination address, and fee payment). Deposits and withdrawals on the asset chains are public, which ties each such transaction to a backend account identifier. Because backend accounts are pseudonymous, public observers can link these transactions to an account but learn little about its internal activity.

Beyond what is public, the Crossroads signing committee sees each withdrawal request it signs, including the requesting account, even if those are not submitted on chain. However, that leakage can itself be minimized by requiring committee nodes to run within trusted enclaves, in which case operators would learn at most that a signature was issued. Our implementation realizes the first guarantee: with Sapphire as the backend, transfers and exchanges within Crossroads are more private, and the committee runs as plain MPC signers (Section~\ref{sec:implementation}). A TEE-based committee node implementation remains future work.

\section{Related Work}\label{sec:related-work}

\paragraph{Cross-chain messaging protocols and bridges.} Cross-chain messaging protocols deliver arbitrary messages from a contract on one chain to a contract on another. Asset bridges are cross-chain messaging protocols specifically for sending assets from one chain to another. We give a non-exhaustive overview of representative ones. Wormhole~\cite{wormhole} and Axelar~\cite{axelar} use validator set committees to attest to cross-chain messages. Wormhole's Token Bridge and Axelar's Satellite (and aggregator Squid) sit on top of these messaging layers to provide token bridging functionality. Hyperlane~\cite{hyperlane} and LayerZero~\cite{zarick2025layerzero} both  let applications configure their own verification logic (i.e. applications can swap in different verifiers, including other cross-chain protocols)---Hyperlane via their Interchain Security Modules, and LayerZero V2 via their Decentralized Verifier Networks. Hyperlane's Warp Routes and LayerZero's Stargate are the canonical asset bridges built on top of these messaging layers Polyhedra Network~\cite{polyhedra} instantiates the zkBridge protocol~\cite{xie2022zkbridge} for both EVM-based chains and Bitcoin. 

\mypara{Cross-chain exchanges.} While bridges transfer a given asset across chains, cross-chain exchanges allow users to swap \emph{between} assets on different chains. Atomic swaps via Hash Time Locked Contracts (HTLCs)~\cite{10.1145/3212734.3212736} are one technique, with subsequent work extending it to arbitrary signature schemes and removing the dependence on on-chain scripting~\cite{thyagarajan2022universal}. Cross-chain AMMs eliminate the need for matching counterparties altogether by pooling liquidity: THORChain~\cite{thorchain} and Chainflip~\cite{chainflip} each operate a dedicated consensus layer that records balances and prices, while threshold-signature-controlled vaults custody the native assets. These systems share Crossroads' backend chain approach but fix a single application at the protocol level: their validators must run nodes for every supported chain, and integrating a new chain is a permissioned protocol decision. Despite the existence of decentralized cross-chain exchanges, centralized exchanges such as Binance, Bybit and Coinbase Exchange dominate the exchange landscape~\cite{coingeckoexchanges}.

\mypara{MPC-based custody.} Threshold signing, used by Crossroads' signing committee, is already widely deployed for digital asset custody. Commercial platforms such as Fireblocks~\cite{fireblocks} secure institutional wallets by sharding signing keys across parties using threshold-ECDSA protocols~\cite{canetti2020uc}, often hardened with TEEs; ZenGo~\cite{zengo} applies two-party ECDSA~\cite{lindell2017fast} to consumer wallets. In all these systems, however, the signing policy is enforced by the provider's off-chain policy engine, opaque to users and not programmable. Crossroads adds a layer of programability via the use of the backend blockchain: signatures are issued only for withdrawals the backend asset contract has authorized, making the custody policy transparent and programmable.

\mypara{Account abstraction.} Account abstraction is today's dominant route to wallet-layer interoperability: ERC-4337~\cite{erc4337} and EIP-7702~\cite{eip7702} make account logic programmable, and \emph{chain-abstraction} systems build on this to let a single wallet control funds across multiple chains. Particle Network~\cite{particle}, for instance, deploys a smart account for the user on each chain, all tied to one key pair. NEAR chain signatures~\cite{nearchainsig} uses an MPC network holding the users keys across multiple chains and, at the account's request, produces the signatures to spend from them. Both give the user a single account, but its assets stay scattered on their native chains, and moving them still requires per-operation bridging. Crossroads achieves the same abstraction with one account on the backend chain, but additionally enables cross-chain operations as ordinary contract calls.

\section{Conclusion}\label{sec:conclusion}

We presented Crossroads, a smart contract layer for chain-abstracted assets that unifies cross-chain asset management on a single backend blockchain. While this presents a step forward towards a more transparent and flexible interoperability framework, it also gives rise to a number of future research directions:

\mypara{Asset-specific policies.} Crossroads currently focuses on cross-chain asset transfers. Many assets, however, carry capabilities tied to their source chains beyond their value: DAO tokens grant voting rights, NFTs unlock token-gated sites or events, etc. As shown in Liquefaction~\cite{austgen2025liquefaction}, key encumbrance can preserve these capabilities for the current owner via per-asset access policies. Extending Crossroads with such policies in a general way remains an open problem.

\mypara{Responsible and fair upgrades.} Crossroads relies on the security of its oracles to safeguard assets. However, an asset chain might hard fork, and its oracle implementation must be upgraded. Designing governance mechanisms that handle oracle upgrades both responsibly and securely is also an open problem.

\ifACM \bibliographystyle{ACM-Reference-Format.bst} \fi
\ifUSENIX \bibliographystyle{plain} \fi
\ifIEEE \bibliographystyle{plain} \fi
\ifLNCS \bibliographystyle{Conferences/LNCS/splncs04.bst} \fi

\bibliography{references} 

\begin{thebibliography}{10}

\bibitem{aave}
Aave.
\newblock Aave: Open source liquidity protocol.
\newblock \url{https://aave.com}, 2026.
\newblock Accessed: 2026-04-24.

\bibitem{across}
{Across Protocol}.
\newblock Across: An intent-based cross-chain bridge.
\newblock \url{https://docs.across.to}, 2026.
\newblock Accessed: 2026-06-11.

\bibitem{uniswapv3}
Hayden Adams, Noah Zinsmeister, Moody Salem, River Keefer, and Dan Robinson.
\newblock Uniswap v3 core.
\newblock \url{https://uniswap.org/whitepaper-v3.pdf}, 2021.
\newblock Accessed: 2026-03-24.

\bibitem{amd_sev_snp}
{Advanced Micro Devices}.
\newblock {AMD SEV-SNP: Strengthening VM Isolation with Integrity Protection
  and More}.
\newblock
  \url{https://www.amd.com/system/files/TechDocs/SEV-SNP-strengthening-vm-isolation-with-integrity-protection-and-more.pdf},
  2020.
\newblock Accessed: 2026-04-28.

\bibitem{aws_nitro}
{Amazon Web Services}.
\newblock {AWS Nitro Enclaves User Guide}.
\newblock
  \url{https://docs.aws.amazon.com/enclaves/latest/user/nitro-enclave.html},
  2024.
\newblock Accessed: 2026-04-28.

\bibitem{austgen2025liquefaction}
James Austgen, Andrés Fábrega, Mahimna Kelkar, Dani Vilardell, Sarah Allen,
  Kushal Babel, Jay Yu, and Ari Juels.
\newblock Liquefaction: Privately liquefying blockchain assets.
\newblock In {\em 2025 IEEE Symposium on Security and Privacy (SP)}, pages
  1493--1511, 2025.

\bibitem{euler}
Michael Bentley and Doug Hoyte.
\newblock Euler whitepaper.
\newblock
  \url{https://github.com/euler-xyz/euler-vault-kit/blob/master/docs/whitepaper.md},
  2024.
\newblock Accessed: 2026-04-24.

\bibitem{boneh2023applied}
Dan Boneh and Victor Shoup.
\newblock {\em A Graduate Course in Applied Cryptography}.
\newblock Draft (Version 0.6), 2023.
\newblock Available at \url{https://toc.cryptobook.us/}.

\bibitem{erc4337}
Vitalik Buterin, Yoav Weiss, Dror Tirosh, Shahaf Nacson, Alex Forshtat, Kristof
  Gazso, and Tjaden Hess.
\newblock {ERC}-4337: Account abstraction using alt mempool.
\newblock \url{https://eips.ethereum.org/EIPS/eip-4337}, 2021.
\newblock Accessed: 2026-06-11.

\bibitem{eip7702}
Vitalik Buterin, Sam Wilson, Ansgar Dietrichs, and Matt Garnett.
\newblock {EIP}-7702: Set code for {EOAs}.
\newblock \url{https://eips.ethereum.org/EIPS/eip-7702}, 2024.
\newblock Accessed: 2026-06-11.

\bibitem{canetti2020uc}
Ran Canetti, Rosario Gennaro, Steven Goldfeder, Nikolaos Makriyannis, and Udi
  Peled.
\newblock {UC} non-interactive, proactive, threshold {ECDSA} with identifiable
  aborts.
\newblock In {\em Proceedings of the 2020 ACM SIGSAC Conference on Computer and
  Communications Security (CCS '20)}, pages 1769--1787. Association for
  Computing Machinery, 2020.

\bibitem{chainalysis2026kelpdao}
{Chainalysis Team}.
\newblock Inside the {KelpDAO} bridge exploit: How {\textasciitilde}\$292
  million in {rsETH} was released against a non-existent burn.
\newblock
  \url{https://www.chainalysis.com/blog/kelpdao-bridge-exploit-april-2026/}, 23
  Apr.~2026.
\newblock Accessed: 2026-06-11.

\bibitem{cheng2019ekiden}
Raymond Cheng, Fan Zhang, Jernej Kos, Warren He, Nicholas Hynes, Noah Johnson,
  Ari Juels, Andrew Miller, and Dawn Song.
\newblock Ekiden: A platform for confidentiality-preserving, trustworthy, and
  performant smart contracts.
\newblock In {\em 2019 IEEE European Symposium on Security and Privacy
  (EuroS\&P)}, pages 185--200, 2019.

\bibitem{chuang2026tee}
Jalen Chuang, Alex Seto, Nicolas Berrios, Stephan van Schaik, Christina Garman,
  and Daniel Genkin.
\newblock Tee. fail: Breaking trusted execution environments via ddr5 memory
  bus interposition.
\newblock In {\em 47th IEEE Symposium on Security and Privacy (IEEE S\&P’26).
  IEEE Computer Society}, 2026.

\bibitem{circle-cctp-finality}
{Circle}.
\newblock Finality and block confirmations.
\newblock Circle CCTP Documentation, 2026.
\newblock Accessed: 2026-06-09.

\bibitem{coinbasebatching}
{Coinbase}.
\newblock Reflections on {Bitcoin} transaction batching.
\newblock
  \url{https://www.coinbase.com/blog/reflections-on-bitcoin-transaction-batching},
  2020.
\newblock Accessed: 2026-06-23.

\bibitem{cexdex2026report}
CoinGecko.
\newblock Cex \& dex trading activity report 2026, 2026.
\newblock Accessed: 2026-03-24.

\bibitem{coingeckoexchanges}
{CoinGecko}.
\newblock Top crypto exchanges ranked by trust score.
\newblock \url{https://www.coingecko.com/en/exchanges}, 2026.
\newblock Accessed: 2026-04-29.

\bibitem{damgrard20fastthreshold}
Ivan Damg\r{a}rd, Thomas~Pelle Jakobsen, Jesper~Buus Nielsen, Jakob~Illeborg
  Pagter, and Michael~B\ae{}ksvang \O{}stergaard.
\newblock Fast threshold ecdsa with honest majority.
\newblock In {\em Security and Cryptography for Networks: 12th International
  Conference, SCN 2020, Amalfi, Italy, September 14–16, 2020, Proceedings},
  page 382–400, Berlin, Heidelberg, 2020. Springer-Verlag.

\bibitem{stablellama}
DefiLlama.
\newblock {DefiLlama: Stablecoins}.
\newblock \url{https://defillama.com/stablecoins}, 2026.
\newblock Accessed: 2026-04-24.

\bibitem{defillamastaking}
DefiLlama.
\newblock {DefiLlama: Staking}.
\newblock \url{https://defillama.com/lsd}, 2026.
\newblock Accessed: 2026-04-24.

\bibitem{morpho}
Mathis~Gontier Delaunay, Paul Frambot, Quentin Garchery, Merlin Egalite, and
  Adrien Tabarly.
\newblock Morpho {B}lue whitepaper.
\newblock
  \url{https://github.com/morpho-org/morpho-blue/blob/main/morpho-blue-whitepaper.pdf},
  2023.
\newblock Accessed: 2026-04-29.

\bibitem{fireblocks}
{Fireblocks}.
\newblock Fireblocks: Digital asset custody and {MPC} wallet infrastructure.
\newblock \url{https://www.fireblocks.com}, 2026.
\newblock Accessed: 2026-06-10.

\bibitem{gast2025counterseveillance}
Stefan Gast, Hannes Weissteiner, Robin~Leander Schr{\"o}der, and Daniel Gruss.
\newblock Counterseveillance: Performance-counter attacks on amd sev-snp.
\newblock In {\em Network and Distributed System Security (NDSS) Symposium
  2025}, 2025.

\bibitem{axelar}
Sergey Gorbunov and Georgios Vlachos.
\newblock Axelar network: Connecting applications with blockchain ecosystems.
\newblock \url{https://docs.axelar.dev/axelar_whitepaper.pdf}, 2021.

\bibitem{chainflip}
Simon Harman.
\newblock Chainflip protocol whitepaper.
\newblock \url{https://assets.chainflip.io/whitepaper.pdf}, 2023.
\newblock Fifth Revision. Accessed: 2026-04-29.

\bibitem{10.1145/3212734.3212736}
Maurice Herlihy.
\newblock Atomic cross-chain swaps.
\newblock In {\em Proceedings of the 2018 ACM Symposium on Principles of
  Distributed Computing}, PODC '18, page 245–254, New York, NY, USA, 2018.
  Association for Computing Machinery.

\bibitem{hyperlane}
{Hyperlane}.
\newblock Hyperlane: Permissionless interoperability for any blockchain.
\newblock \url{https://docs.hyperlane.xyz}, 2023.
\newblock Accessed: 2026-04-29.

\bibitem{intel_sgx}
{Intel Corporation}.
\newblock {Intel Software Guard Extensions (Intel SGX) Developer Reference}.
\newblock
  \url{https://www.intel.com/content/www/us/en/developer/tools/software-guard-extensions/overview.html},
  2023.
\newblock Accessed: 2026-04-28.

\bibitem{komlo2020frost}
Chelsea Komlo and Ian Goldberg.
\newblock Frost: Flexible round-optimized schnorr threshold signatures.
\newblock In {\em Selected Areas in Cryptography: 27th International
  Conference, Halifax, NS, Canada (Virtual Event), October 21-23, 2020, Revised
  Selected Papers}, page 34–65, Berlin, Heidelberg, 2020. Springer-Verlag.

\bibitem{krause20251}
David Krause.
\newblock The \$1.4 billion {Bybit} hack: Cybersecurity failures and the risks
  of cryptocurrency deregulation.
\newblock {\em International Journal of Cryptocurrency Research},
  5(1):10--51483, 2025.

\bibitem{oasis2026privana}
Nino Kutnjak.
\newblock Privana: A practical liquefaction implementation, April 2026.

\bibitem{lindell2017fast}
Yehuda Lindell.
\newblock Fast secure two-party {ECDSA} signing.
\newblock In {\em Advances in Cryptology -- CRYPTO 2017}, pages 613--644.
  Springer, 2017.

\bibitem{10.1145/501983.502017}
Silvio Micali, Kazuo Ohta, and Leonid Reyzin.
\newblock Accountable-subgroup multisignatures: extended abstract.
\newblock In {\em Proceedings of the 8th ACM Conference on Computer and
  Communications Security}, CCS '01, page 245–254, New York, NY, USA, 2001.
  Association for Computing Machinery.

\bibitem{airswap}
{Michael Oved, Don Mosites}.
\newblock Swap: A peer-to-peer protocol for trading ethereum tokens, 2017.

\bibitem{near_threshold_signatures_2026}
{NEAR}.
\newblock Threshold signatures.
\newblock
  \url{https://github.com/near/mpc/tree/main/crates/threshold-signatures},
  2026.
\newblock Accessed: 2026-04-29.

\bibitem{nearchainsig}
{NEAR Foundation}.
\newblock Chain signatures.
\newblock \url{https://docs.near.org/chain-abstraction/chain-signatures}, 2026.
\newblock Accessed: 2026-06-11.

\bibitem{oasis2023sapphire}
{Oasis Protocol Foundation}.
\newblock Sapphire paratime, 2026.
\newblock Accessed Jan.~2026.

\bibitem{particle}
{Particle Network}.
\newblock Particle network: Chain abstraction via universal accounts.
\newblock \url{https://whitepaper.particle.network}, 2024.
\newblock Accessed: 2026-06-11.

\bibitem{polyhedra}
{Polyhedra Network}.
\newblock Polyhedra network zkbridge.
\newblock \url{https://docs.zkbridge.com/}, 2023.
\newblock Accessed: 2026-04-29.

\bibitem{rangoAPI}
{Rango}.
\newblock {Choosing the Right API}.
\newblock
  \url{https://docs.rango.exchange/api-integration/choosing-the-right-api#main-api-multi-step-txs}.
\newblock Accessed: 2026-06-10.

\bibitem{RMPocalypse2025}
Benedict Schlüter and Shweta Shinde.
\newblock Rmpocalypse: How a catch-22 breaks amd sev-snp.
\newblock In {\em Proceedings of the 2025 on ACM SIGSAC Conference on Computer
  and Communications Security}, CCS '25. Association for Computing Machinery,
  2025.

\bibitem{solana_durable_nonces}
{Solana Foundation}.
\newblock Durable nonces.
\newblock \url{https://solana.com/docs/core/transactions/durable-nonces}, 2026.
\newblock Accessed: 2026-04-28.

\bibitem{tas2023bitcoin}
Ertem~Nusret Tas, David Tse, Fangyu Gai, Sreeram Kannan, Mohammad~Ali
  Maddah-Ali, and Fisher Yu.
\newblock Bitcoin-enhanced proof-of-stake security: Possibilities and
  impossibilities.
\newblock In {\em 2023 IEEE Symposium on Security and Privacy (SP)}, pages
  126--145. IEEE, 2023.

\bibitem{tas2022babylon}
Ertem~Nusret Tas, David Tse, Fisher Yu, and Sreeram Kannan.
\newblock Babylon: Reusing bitcoin mining to enhance proof-of-stake security.
\newblock {\em arXiv preprint arXiv:2201.07946}, 2022.

\bibitem{tetherprotocols}
{Tether}.
\newblock Supported protocols and integration guidelines.
\newblock \url{https://tether.to/en/supported-protocols/}, 2026.
\newblock Accessed: 2026-04-29.

\bibitem{thorchain}
{THORChain}.
\newblock {THORChain}: A decentralised liquidity network.
\newblock
  \url{https://github.com/thorchain/Resources/blob/master/Whitepapers/whitepaper-en.md},
  2020.
\newblock Accessed: 2026-04-29.

\bibitem{thyagarajan2022universal}
Sri~AravindaKrishnan Thyagarajan, Giulio Malavolta, and Pedro Moreno-Sanchez.
\newblock Universal atomic swaps: Secure exchange of coins across all
  blockchains.
\newblock In {\em 2022 IEEE symposium on security and privacy (SP)}, pages
  1299--1316. IEEE, 2022.

\bibitem{erc20standard}
Fabian Vogelsteller and Vitalik Buterin.
\newblock {ERC-20: Token Standard}, November 2015.

\bibitem{wilke2024tdxdown}
Luca Wilke, Florian Sieck, and Thomas Eisenbarth.
\newblock Tdxdown: Single-stepping and instruction counting attacks against
  intel tdx.
\newblock In {\em Proceedings of the 2024 on ACM SIGSAC Conference on Computer
  and Communications Security}, pages 79--93, 2024.

\bibitem{0xprotocol}
{Will Warren, Amir Bandeali}.
\newblock 0x: An open protocol for decentralized exchange on the ethereum
  blockchain, 2017.

\bibitem{wormhole}
{Wormhole Foundation}.
\newblock Wormhole messaging.
\newblock \url{https://wormhole.com/docs/products/messaging/overview/}, 2022.
\newblock Accessed: 2026-04-29.

\bibitem{xie2022zkbridge}
Tiancheng Xie, Jiaheng Zhang, Zerui Cheng, Fan Zhang, Yupeng Zhang, Yongzheng
  Jia, Dan Boneh, and Dawn Song.
\newblock zkbridge: Trustless cross-chain bridges made practical.
\newblock In {\em Proceedings of the 2022 ACM SIGSAC Conference on Computer and
  Communications Security}, pages 3003--3017, 2022.

\bibitem{yuan2025ciphersteal}
Yuanyuan Yuan, Zhibo Liu, Sen Deng, Yanzuo Chen, Shuai Wang, Yinqian Zhang, and
  Zhendong Su.
\newblock Ciphersteal: Stealing input data from tee-shielded neural networks
  with ciphertext side channels.
\newblock In {\em 2025 IEEE Symposium on Security and Privacy (SP)}, pages
  4136--4154. IEEE, 2025.

\bibitem{zarick2025layerzero}
Ryan Zarick, Bryan Pellegrino, Isaac Zhang, Thomas Kim, and Caleb Banister.
\newblock Layerzero, 2025.

\bibitem{zengo}
{Zengo}.
\newblock Zengo: The {MPC} crypto wallet.
\newblock \url{https://zengo.com}, 2026.
\newblock Accessed: 2026-06-10.

\bibitem{zhang2024cross}
Mengya Zhang, Xiaokuan Zhang, Yinqian Zhang, and Zhiqiang Lin.
\newblock Cross-chain bridges: Attack taxonomy, defenses, and open problems.
\newblock In {\em RAID 2024}, 2024.

\end{thebibliography}
\ifSP \appendices
\crefalias{section}{appendix}
\else
\appendix \fi

\ifCCS
\input{Sections/ai_use}
\input{Sections/appendix/open_science}
\fi
\section*{Ethics Considerations}
We acknowledge that it is possible to deploy Crossroads in a privacy-preserving manner such that user transactions cannot be easily tracked across chains. We have not deployed Crossroads this way ourselves, but we have considered how this could be done responsibly. In particular, our use of TEEs on top of the signing committee members can be extended to allow for conditional release of information to regulators in special circumstances according to a policy (as opposed to broad access to all activity).

\section{Other Applications}\label{appdx:other_applications}

This section describes additional applications for Crossroads.

\subsection{Wallets and Liquidity}

\mypara{Chain-agnostic wallet.} Any asset from any blockchain integrated into Crossroads can be managed from a single backend blockchain wallet. When making purchases in an asset the wallet doesn't already own, frontend software can abstract away the swap step to provide a payment experience denominated in the user's asset of choice. Because every user's holdings across all integrated chains are recorded in a single backend ledger, wallets and block explorers can display a user's complete cross-chain portfolio and activity from one source, rather than querying and aggregating a separate full node for each chain.

\mypara{Universal receive addresses.} Crossroads can support \textit{receiving} assets from arbitrary integrated chains and automatically converting the payments to the asset of the recipient's choice. Most donation recipients and merchants only accept a few different blockchain asset types. Using custom deposit addresses as described in~\Cref{subsec:optimized-deposits}, a recipient can instead designate a preferred asset (e.g., Ether). Crossroads can derive receive addresses for all integrated blockchains, binding these addresses to a contract that automatically swaps received funds for the recipient's preferred asset.

\subsection{Universal Assets}

\mypara{Universal stablecoins.} Stablecoins are a class of cryptocurrencies designed to maintain a stable value relative to a reference asset, typically USD. The market has grown rapidly, with aggregate capitalization exceeding \$300 billion, of which roughly \$250 billion is concentrated in USDC and USDT, the two largest stablecoins~\cite{stablellama}. Although both are deployed on multiple chains,  chain support is determined by a centralized issuer and therefore remains selective: USDT, for example, is only available on 13 blockchains despite the existence of over 150 active ones~\cite{tetherprotocols}.

Crossroads enables the construction of a universal stablecoin: a stablecoin issued on the backend blockchain that can be spent on any smart-contract-enabled chain integrated with Crossroads. To extend support to a new chain, anyone can deploy a smart contract on that chain that mints stablecoins to a user upon presentation of a valid withdrawal signature from the Crossroads address.

\mypara{Cross-chain staking.} Over \$34 billion is currently staked on on-chain protocols~\cite{defillamastaking}, but this capital only secures protocols on the chain where it is held. Large-liquidity chains benefit, while emerging chains struggle to attract stake sufficient to secure their protocols. Babylon addresses this problem by enabling use of BTC for staking against equivocation on other chains~\cite{tas2022babylon,tas2023bitcoin}. Crossroads offers an alternative, enabling a cross-chain staking model using assets \textit{from any integrated chain} for \textit{any malfeasance observable by an oracle}.

Concretely, a staking contract on the backend blockchain holds a user's wrapped assets as collateral and receives oracle updates from the target chain. When the oracle relays evidence of misbehavior, the responsible actor's collateral is slashed.

\mypara{Cross-chain lending and borrowing.} Lending and borrowing protocols are widespread in blockchain settings, commonly used to short assets or arbitrage price differences across markets. Crossroads enables assets from any integrated chain to serve as collateral to borrow any other integrated asset, making atomic cross-chain shorts and arbitrage practical. Standard lending contracts such as Aave~\cite{aave} or Euler~\cite{euler} can be deployed on the backend blockchain to enable this functionality.

\subsection{Privacy and Compliance}

\mypara{Private management of public assets.} Institutional asset management requires confidentiality: portfolios, trade timing, and counterparty relationships are commercially sensitive. Yet most crypto liquidity lives on public blockchains, where every transaction is visible. Today, institutions rely on off-chain custodians or centralized exchanges, exposing themselves to custody fees and losing access to smart contract functionality.

Crossroads instead enables institutional management of public blockchain assets using a fully-featured private backend blockchain. Additionally, Crossroads enables deployment of funds in decentralized finance protocols even while they are custodied in Crossroads, allowing institutions to take advantage of the pooled liquidity of public blockchains while privately enforcing access controls and ownership changes.

\mypara{On-chain compliance.} Compliance policies in Crossroads can be implemented as smart contracts and layered directly onto exchange transactions, for example by enforcing KYC-based whitelisting that restricts participation to vetted counterparties. Adherence to these policies is then publicly verifiable on-chain, rather than enforced by an opaque centralized authority, and the policies compose with the private applications described above.

\section{Additional Proofs}\label{appdx:proofs}

\begin{proof}[Proof of Theorem~\ref{thm:soundness}]
We proceed in three parts. First, we establish an invariant showing that the smart contract accurately tracks the users $\mathsf{Net}$ balance. Second, we show the encumbered address holds sufficient native assets on the integrated chain to support withdrawals by any user. Third, we construct a sequence of $\mathsf{acct}$-issuable operations such that any valid transcript extending $\vec{\sigma}$ that contains the sequence in order includes a valid broadcast of the requested withdrawal transaction.

\paragraph{Part 1: Internal state accuracy.}
Let $N_{\mathsf{acct}} = \mathsf{Net}(\mathsf{acct}, \mathsf{chain}, \vec{\sigma})$. We assert that the contract's internal state variables, $\mathsf{bal}$ and $\mathsf{PW}$, maintain the invariant:
\[
\mathsf{bal}[(\mathsf{acct}, \mathsf{chain})] + \mathsf{PW}[\mathsf{acct}, \mathsf{chain}] = N_{\mathsf{acct}}
\]
We prove this by induction on the length of $\vec{\sigma}$.
\begin{itemize}
    \item \textbf{Base Case:} Before any operations are executed, $\vec{\sigma} = \emptyset$. $D \gets \emptyset$, $\mathsf{bal} = 0$, $\mathsf{PW} = 0$, and $N_{\mathsf{acct}} = 0$. The invariant holds trivially.
    \item \textbf{Inductive Step:} Assume the invariant holds for a prefix $\vec{\sigma}_k$. We case-split on the $k+1$-th operation:
    \begin{itemize}
        \item \textsf{Deposit}: If credited to $\mathsf{acct}$, the contract increments $\mathsf{bal}[(\mathsf{acct}, \mathsf{chain})]$ by $\mathsf{amt}$. Simultaneously, this operation enters $\mathsf{Dep}(\mathsf{acct})$, increasing $N_{\mathsf{acct}}$ by $\mathsf{amt}$. The invariant is preserved.
        \item \textsf{Transfer}: If sent by $\mathsf{sdr}$ to $\mathsf{rcv}$, $\mathsf{bal}[(\mathsf{sdr}, \mathsf{chain})]$ decreases by $\mathsf{amt}$ and $\mathsf{bal}[(\mathsf{rcv}, \mathsf{chain})]$ increases by $\mathsf{amt}$. On the right side of the equation, $N_{\textsf{sdr}}$ decreases via $\mathsf{Snd}$ by $\mathsf{amt}$ and $N_{\textsf{rcv}}$ increases via $\mathsf{Rcv}$ by $\mathsf{amt}$. The invariant is preserved for both.
        \item \textsf{LockWithdrawal}: Initiated by $\mathsf{acct}$, this decreases $\mathsf{bal}$ by $\mathsf{amt}$ and increases $\mathsf{PW}$ by $\mathsf{amt}$. The sum $\mathsf{bal} + \mathsf{PW}$ remains unchanged, matching $N_{\mathsf{acct}}$ which is unaffected by this internal state transition. The invariant is preserved.
        \item \textsf{ConfirmWithdrawal}: Attributed to $\textsf{acct}$, this deducts $(\textsf{amt} + \textsf{fee})$ from $\textsf{PW}$. This finalized withdrawal enters $\textsf{With}(\textsf{acct})$, reducing $N_{\textsf{acct}}$ by the same $\textsf{amt} + \textsf{fee}$. The invariant is preserved.
        \item \textsf{Other}: Other operations modify none of the relevant variables, leaving the invariant intact.
    \end{itemize}
\end{itemize}
Thus the internal ledger accurately tracks the abstract $\mathsf{Net}$.

\paragraph{Part 2: On-chain balance consistency.}
Let $\textsf{OnChain}(\mathsf{chain})$ denote the native balance of the crossroads address on $\mathsf{chain}$. 
By Oracle Soundness, every recorded $\textsf{Deposit}$ and $\textsf{ConfirmWithdrawal}$ corresponds to a finalized on-chain transfer to or from the address. Assuming every spend from the crossroads address is reflected as a $\textsf{ConfirmWithdrawal}$ in $\vec{\sigma}$ (which we establish at the end of this proof), unrecorded finalized deposits may only inflate $\textsf{OnChain}(\mathsf{chain})$ above what $\vec{\sigma}$ accounts for. Hence:
\[
\textsf{OnChain}(\mathsf{chain}) \;\ge\; \sum_{\mathsf{acct} \in \mathcal{U}} N_{\mathsf{acct}}
\]
where $\textsf{Transfer}$ operations cancel in the sum over $\mathcal{U}$. Since each $N_{\mathsf{acct}} \ge 0$, for every $\mathsf{acct}$ we have $\textsf{OnChain}(\mathsf{chain}) \ge N_{\mathsf{acct}}$.

\paragraph{Part 3: Fulfilling the soundness guarantee.}
Let $\mathsf{acct}, \mathsf{chain}, \mathsf{addr}_{\mathsf{dest}}, \mathsf{amt}, \mathsf{fee}$ be as in Definition~\ref{def:soundness}, so that $\mathsf{Net}(\mathsf{acct}, \mathsf{chain}, \vec{\sigma}) \ge \mathsf{amt} + \mathsf{fee}$ and $\mathsf{bal}[(\mathsf{acct}, \eta)] \ge r$. By Part 1, we know $\mathsf{bal}[(\mathsf{acct}, \mathsf{chain})] + \mathsf{PW}[\mathsf{acct}, \mathsf{chain}] \ge \mathsf{amt} + \mathsf{fee}$. By Part 2, $\textsf{OnChain}(\mathsf{chain}) \ge N_{\mathsf{acct}} \ge \mathsf{amt} + \mathsf{fee}$, so the Crossroads-controlled address on $\mathsf{chain}$ holds sufficient native assets to back this withdrawal. Let $\ell = \mathsf{PW}[\mathsf{acct}, \mathsf{chain}]$ denote the amount already locked by $\mathsf{acct}$ on $\mathsf{chain}$, and set $\Delta = \max(0, (\mathsf{amt} + \mathsf{fee}) - \ell)$. If $\Delta > 0$, then since $\mathsf{bal}[(\mathsf{acct}, \mathsf{chain})] + \ell \ge \mathsf{amt} + \mathsf{fee}$, we have $\mathsf{bal}[(\mathsf{acct}, \mathsf{chain})] \ge \Delta$.

$\mathsf{acct}$ then executes the following sequence:

\begin{enumerate}
    \item \textbf{Locking Funds:} $\mathsf{acct}$ calls $\textsf{LockWithdrawal}(\mathsf{chain}, \Delta)$ on the $\textsf{crsrds}$ contract. Both of the contract's assertions hold: $\mathsf{bal}[(\mathsf{acct}, \mathsf{chain})] \ge \Delta$ as argued above, and $\mathsf{bal}[(\mathsf{acct}, \eta)] \ge r$ by assumption. The call shifts $\Delta$ from $\mathsf{bal}$ into $\mathsf{PW}[\mathsf{acct}, \mathsf{chain}]$ and escrows the prover reward, adding $r$ to $\mathsf{RW}[\mathsf{acct}]$. (The call is made even when $\Delta = 0$: it then moves no chain balance, but still establishes the escrowed reward that \textsf{CanSign} requires.) After this step, $\mathsf{PW}[\mathsf{acct}, \mathsf{chain}] \ge \mathsf{amt} + \mathsf{fee}$ and $\mathsf{RW}[\mathsf{acct}] \ge r$.

    \item \textbf{Signature Generation:} The user invokes $\Pi_{\textsf{SC}}^{\textsf{TSS}}.\textsf{Sign}(\mathsf{chain}, \mathsf{amt}, \mathsf{fee}, \mathsf{addr}_{\mathsf{dest}})$. Each honest signer $s_j$ independently verifies $\textsf{crsrds}.\textsf{CanSign}(\mathsf{acct}, \mathsf{chain}, \mathsf{amt}, \mathsf{fee}) = 1$, which holds because $\mathsf{PW}[\mathsf{acct}, \mathsf{chain}] \ge \mathsf{amt} + \mathsf{fee}$ and $\mathsf{RW}[\mathsf{acct}] \ge r$ after Step~1, then fetches the current $\mathsf{epoch}_0$ via $\textsf{crsrds}.\textsf{GetEpoch}$, constructs $\mathsf{tx}$ with this epoch, recipient $\mathsf{addr}_{\mathsf{dest}}$, and value $\mathsf{amt}$, and produces a signature share $\sigma^{(j)} \gets \textsf{TSS}.\mathsf{Sign}(\textsf{sk}^{(j)}_{\mathsf{chain}}, \mathsf{tx})$. Since fewer than $t$ signers are corrupted and $n \ge 2t-1$, at least $t$ honest signers produce valid shares, so $\mathsf{acct}$ collects enough shares to run $\textsf{TSS}.\mathsf{Combine}$ successfully, obtaining $\sigma$.

    \item \textbf{On-Chain Settlement:} $\mathsf{acct}$ broadcasts $(\mathsf{tx}, \sigma)$ to $\mathsf{chain}$. By construction in Step~2, $\mathsf{tx}$ is well-formed with $\mathsf{recipient}(\mathsf{tx}) = \mathsf{addr}_{\mathsf{dest}}$ and $\mathsf{value}(\mathsf{tx}) = \mathsf{amt}$, the signature $\sigma$ verifies under $\textsf{pk}_{\mathsf{chain}}$, and $\mathsf{tx}$ carries the epoch $\mathsf{epoch}_0 = \textsf{crsrds}.\textsf{GetEpoch}(\mathsf{chain})$. By Part~2, the encumbered address on $\mathsf{chain}$ holds at least $\mathsf{amt} + \mathsf{fee}$ in native assets, sufficient to cover $\mathsf{value}(\mathsf{tx}) + \mathsf{fee}(\mathsf{tx})$. Thus $\textsf{Valid}_{\mathsf{chain}}(\mathsf{tx}, \mathsf{st}_{\mathsf{chain}}) = 1$ whenever $\mathsf{epoch}_0$ still matches the encumbered address's on-chain epoch at the moment of broadcast.
\end{enumerate}

Because $\mathsf{acct}$ could unilaterally execute and authorize this valid sequence of operations, the system satisfies Definition~\ref{def:soundness}.

A competing withdrawal settling first would advance that on-chain epoch, superseding $\mathsf{tx}$; by single-choice inclusion (Section~\ref{sec:requirements}) the superseded $\mathsf{tx}$ becomes permanently invalid while $\mathsf{acct}$'s pending-withdrawal balance $\mathsf{PW}[\mathsf{acct}, \mathsf{chain}]$ is untouched, so $\mathsf{acct}$ re-fetches the new epoch via $\textsf{crsrds}.\textsf{GetEpoch}$, requests a fresh signature as in Step~2 (the $\textsf{CanSign}$ check still passes, since $\mathsf{PW}$ and $\mathsf{RW}$ are unchanged), and rebroadcasts. Each such competing withdrawal is itself a finalized spend from the encumbered address and therefore, by fee replacement (Section~\ref{sec:requirements}), pays a positive on-chain cost $c > 0$; moreover, to win the race against $\mathsf{acct}$ it must outbid $\mathsf{acct}$'s fee. A resource-bounded adversary with budget $B$ can thus finalize at most $k \le B / c$ competing withdrawals, after which no adversarial withdrawal supersedes $\mathsf{tx}$ and $\mathsf{acct}$'s broadcast satisfies $\textsf{Valid}_{\mathsf{chain}}(\mathsf{tx}, \mathsf{st}_{\mathsf{chain}}) = 1$. Hence within $k+1$ attempts $\mathsf{acct}$ broadcasts a valid withdrawal transaction, as required by Definition~\ref{def:soundness}.

It remains to discharge the assumption used in Part~2 that every spend from the encumbered address corresponds to a $\textsf{ConfirmWithdrawal}$ entry in $\vec{\sigma}$. An unrecorded spend would require a valid signature on some transaction $\mathsf{tx}'$ that was \emph{not} authorized via $\textsf{crsrds}.\textsf{CanSign}$; were such a transaction to settle on $\mathsf{chain}$, it would drain the encumbered address below $\sum_{\mathsf{acct}} N_{\mathsf{acct}}$ and break the on-chain consistency that $\mathsf{acct}$'s withdrawal in Part~3 relies on. In $\Pi_{\textsf{SC}}^{\textsf{TSS}}$, every honest signer checks $\textsf{CanSign}$ before producing a share, so no honest share on an unauthorized transaction is ever issued by the signing committee. An adversary corrupting fewer than $t$ signers therefore holds at most $t-1$ valid shares on $\mathsf{tx}'$. By the security of $\Pi_{\textsf{DKG}}$, the joint distribution of the adversary's view and the honest parties' key shares is indistinguishable from one produced by an ideal trusted dealer, so producing a signature that satisfies $\textsf{TSS}.\mathsf{Verify}$ reduces to an EUF-CMA forgery against $\textsf{TSS}$, which occurs with negligible probability.
\end{proof}

\section{Denial-of-service attacks}
\label{subsec:dos}

Since Crossroads withdrawals on a given chain originate from the same encumbered address, an adversary can attempt to throttle the system by repeatedly initiating withdrawals. By submitting back-to-back withdrawal transactions carrying high fees, the adversary forces honest users to either match those fees to compete for inclusion in subsequent epochs or to wait until the adversary runs out of money. With a sufficiently funded adversary, this can stall withdrawals or make them prohibitively expensive.

We mitigate this attack with two of the mechanisms introduced in~\Cref{sec:faster-withdrawals}. First, with an intent-based withdrawal framework, users do not have to go through the encumbered address at all: when a third-party filler is willing to advance funds on the destination chain in exchange for the user's Crossroads tokens, the user's withdrawal completes via the filler's address rather than competing for the encumbered account's epochs. Second, where the chain supports it, parallelizing withdrawals removes the single-epoch bottleneck altogether. UTXO-based chains naturally support parallel spends through multi-output transactions, while smart-contract-enabled chains can multiplex withdrawals through a dispatcher contract.

\section{Signing Latency Evaluation}
\label{appdx:signing-latency}

\Cref{tab:crossroads-asset-chain-cost} compares Crossroads deposits and withdrawals to ordinary native transfers on each asset chain. The overhead of the transaction binding adds an extra 31 virtual bytes on Bitcoin and 320 extra gas on Ethereum, with no additional cost on Solana.

\begin{table}[t]
	\centering
	\setlength{\tabcolsep}{0.80em}
	\renewcommand{\arraystretch}{1.10}
	\begin{tabular}{@{}c@{\hspace{0.85em}}|@{\hspace{1.10em}}lll@{}}
		\Xhline{0.8pt}
		Asset chain & Native transfer & \makecell[c]{Crossroads\\deposit} & \makecell[c]{Crossroads\\withdrawal} \\
		\hline
		\multirow{2}{*}{Ethereum}
		& \$0.009073 & \$0.009073 & \$0.009211 \\[0.04em]
		& (21{,}000 gas) & (21{,}000 gas) & (21{,}320 gas) \\[0.58em]
		\multirow{2}{*}{Bitcoin}
		& \$0.3473 & \$0.4236 & \$0.4236 \\[0.04em]
		& (141 vB) & (172 vB) & (172 vB) \\[0.58em]
		\multirow{2}{*}{Solana}
		& \$0.000320 & \$0.000320 & \$0.000320 \\[0.08em]
		& (1 sig) & (1 sig) & (1 sig) \\
		\Xhline{0.8pt}
	\end{tabular}
	\caption{Asset chain transaction costs in USD, comparing ordinary transfers on each chain to
		Crossroads deposit and withdrawal transactions. Transaction fee prices are from June 10, 2026:
		0.27 gwei on Ethereum, 4 satoshis/vB on Bitcoin, and the base fee of 5,000 lamports on Solana.}
	\label{tab:crossroads-asset-chain-cost}
\end{table}

~\Cref{tab:impl-signing-latency} reports the threshold ECDSA signing latency in our prototype by committee size and round-trip latency. We emulate the network latency between processes on the same machine to give precise answers, and these measurements use all $n$ members ($t = n$). Threshold ECDSA is the slower of the two signature schemes.

\begin{table}[t]
	\centering
	\scalebox{0.9}{%
	\begin{tabular}{c|r|r|r|r|r|r|r}
		Network \\ Latency\\/ Committee \\ Size & 0ms & 10ms & 50ms & 100ms & 250ms & 500ms & 1000ms \\\hline
		3  & 0.01 & 0.10 & 0.42 & 0.82 & 2.03 & 4.04 & 8.04 \\
		5  & 0.02 & 0.17 & 0.75 &  1.44 & 3.56 & 7.07 & 14.07 \\
		7  & 0.04 & 0.25 & 1.06 & 2.08 & 5.07 & 10.09 & 20.10 \\
		9  & 0.05 & 0.32 & 1.40 & 2.71 & 6.62 & 13.13 & 26.12 \\
		11 & 0.08 & 0.40 & 1.71 & 3.34 & 8.15 & 16.15 & 32.17 \\
		13 & 0.16 & 0.51 & 2.06 & 3.97 & 9.68 & 19.21 & 38.20 \\
		15 & 0.23 & 0.59 & 2.40 & 4.61 & 11.22 & 22.24 & 44.25 \\
	\end{tabular}
	}
	\caption{ECDSA signing latency measured in seconds, varied by signing committee size and round-trip connection latency between committee members. Each number is the average of five trials.}
	\label{tab:impl-signing-latency}
\end{table}

\end{document}